\documentclass[11pt]{article}
\usepackage[margin=1in]{geometry}

\usepackage{enumitem}
\usepackage[utf8]{inputenc}
\usepackage{amsmath,amssymb,tabu}
\usepackage{color}
\usepackage{calc}
\usepackage{tikz}
\usetikzlibrary{positioning,arrows,decorations.pathreplacing,shapes}
\usetikzlibrary{calc}
\usepackage{multirow}
\usepackage{boxedminipage}
\usepackage{xifthen}
\usepackage{tabularx}
\usepackage{xcolor}
\usepackage{setspace}
\usepackage{authblk}

\usepackage{mathpazo}

\usepackage{caption, subcaption}

\usepackage{algorithm}
\usepackage{algpseudocode}
\usepackage{hyperref}
\usepackage{cleveref}

\newcommand{\ignore}[1]{}
\newlength{\bibitemsep}
\newlength{\bibparskip}
\let\oldthebibliography\thebibliography
\renewcommand\thebibliography[1]{%
	\oldthebibliography{#1}%
	\setlength{\parskip}{\bibitemsep}%
	\setlength{\itemsep}{\bibparskip}%
}

\newcommand{\R}{\mathbb{R}}

\newcommand{\M}{\mathcal{M}}
\newcommand{\I}{\mathcal{I}}

\newcommand{\RR}{\mathbb{R}}

\DeclareMathOperator{\vol}{vol}
\DeclareMathOperator{\perm}{perm}
\newcommand{\del}{\backslash}
\DeclareMathOperator{\Span}{Span}
\newcommand\bmat[1]{\begin{bmatrix} #1 \end{bmatrix}}

\newtheorem{claim}{Claim}[section]

\newtheorem{theorem}{Theorem}
\newtheorem{corollary}{Corollary}
\newtheorem{observation}{Observation}
\newtheorem{definition}{Definition}
\newtheorem{lemma}{Lemma}
\newtheorem{remark}{Remark}
\newenvironment{proof}[1][]{\par \noindent {\bf Proof #1}\ }{\hfill$\Box$\par \vspace{11pt}}

\newcommand{\norm}[1]{\left\lVert#1\right\rVert}
\newcommand{\sym}{\mathrm{sym}}

\newlength{\RoundedBoxWidth}
\newsavebox{\GrayRoundedBox}
\newenvironment{GrayBox}[1]%
   {\setlength{\RoundedBoxWidth}{.93\textwidth}
    \def\boxheading{#1}
    \begin{lrbox}{\GrayRoundedBox}
       \begin{minipage}{\RoundedBoxWidth}}%
   {   \end{minipage}
    \end{lrbox}
    \begin{center}
    \begin{tikzpicture}%
       \node(Text)[draw=black!80,fill=white,rounded corners,%
             inner sep=2ex,text width=\RoundedBoxWidth]%
             {\usebox{\GrayRoundedBox}};
        \coordinate(x) at (current bounding box.north west);
        \node [draw=white,rectangle,inner sep=3pt,anchor=north west,fill=white]
        at ($(x)+(6pt,.75em)$) {\boxheading};
    \end{tikzpicture}
    \end{center}}

\newenvironment{defproblemx}[2][]{\noindent\ignorespaces%
                                \FrameSep=6pt%
                                \parindent=0pt%
                \vspace*{-1.5em}
                \ifthenelse{\isempty{#1}}{%
                  \begin{GrayBox}{\textsc{#2}}%
                }{%
                  \begin{GrayBox}{\textsc{#2} parameterized by~{#1}}%
                }
                \begin{tabular*}{\textwidth}{@{\hspace{.1em}} >{\itshape} p{1.8cm} p{0.8\textwidth} @{}}%
            }{
                \end{tabular*}%
                \end{GrayBox}%
                \ignorespacesafterend
            }

\usepackage{rotating}
\usepackage{framed}

\usepackage[textsize=tiny,textwidth=2cm,color=green!50!gray]{todonotes} 

\newcommand*\samethanks[1][\value{footnote}]{\footnotemark[#1]}

\title{Determinant Maximization via Matroid Intersection Algorithms}

\author[1]{Adam Brown\thanks{ajmbrown@gatech.edu, aladdha6@gatech.edu, msingh94@gatech.edu}}
\author[1]{Aditi Laddha \samethanks[1]}
\author[2]{Madhusudhan Pittu\thanks{mpittu@andrew.cmu.edu, ptetali@cmu.edu}}
\author[1]{Mohit Singh\samethanks[1]}
\author[2]{Prasad Tetali\samethanks[2]}

\affil[1]{Georgia Institute of Technology.}
\affil[2]{Carnegie Mellon University.}

\begin{document}

\maketitle

\begin{abstract}
Determinant maximization problem gives a general framework that models problems arising in as diverse fields as statistics~\cite{pukelsheim2006optimal}, convex geometry~\cite{Khachiyan1996}, fair allocations\linebreak~\cite{anari2016nash}, combinatorics~\cite{AnariGV18}, spectral graph theory~\cite{nikolov2019proportional}, network design, and random processes~\cite{kulesza2012determinantal}. In an instance of a determinant maximization problem, we are given a collection of vectors $U=\{v_1,\ldots, v_n\} \subset \RR^d$, and a goal is to pick a subset $S\subseteq U$ of given vectors to maximize the determinant of the matrix $\sum_{i\in S} v_i v_i^\top $. Often, the set $S$ of picked vectors must satisfy additional combinatorial constraints such as cardinality constraint $\left(|S|\leq k\right)$ or matroid constraint ($S$ is a basis of a matroid defined on the vectors).

In this paper, we give a polynomial-time deterministic algorithm that returns a $r^{O(r)}$\linebreak-approximation for any matroid of rank $r\leq d$. This improves previous results that give $e^{O(r^2)}$-approximation algorithms relying on $e^{O(r)}$-approximate \emph{estimation} algorithms~\cite{NikolovS16,anari2017generalization,AnariGV18,madan2020maximizing} for any $r\leq d$. All previous results use convex relaxations and their relationship to stable polynomials and strongly log-concave polynomials. In contrast, our algorithm builds on combinatorial algorithms for matroid intersection, which iteratively improve any solution by finding an \emph{alternating negative cycle} in the \emph{exchange graph} defined by the matroids. While the $\det(.)$ function is not linear, we show that taking appropriate linear approximations at each iteration suffice to give the improved approximation algorithm.
\end{abstract}

\section{Introduction}
Determinant maximization problem gives a general framework that models problems arising in as diverse fields as statistics~\cite{pukelsheim2006optimal}, convex geometry~\cite{Khachiyan1996}, fair allocations~\cite{anari2016nash}, combinatorics~\cite{AnariGV18}, spectral graph theory~\cite{nikolov2019proportional}, network design and random processes~\cite{kulesza2012determinantal}. In an instance of a determinant maximization problem, we are given a collection of vectors $U=\{v_1,\ldots, v_n\} \subset \RR^d$, and a goal is to pick a subset $S\subseteq U$ of given vectors to maximize the determinant of the matrix $\sum_{i\in S} v_i v_i^\top $. Additionally, the set $S$ of picked vectors must satisfy additional combinatorial constraints such as cardinality constraint $\left(|S|\leq k\right)$ or matroid constraint ($S$ is a basis of a matroid defined on the vectors).

Apart from its modeling strength, from a technical perspective, determinant maximization has brought interesting connections between areas such as combinatorial optimization, convex analysis, geometry of polynomials, graph sparsification and complexity of permanent and other counting problems~\cite{allen2017near,anari2016nash,anari2017generalization,Khachiyan1996}.

\paragraph{Applications.} Observe that when the number of vectors picked is exactly $d$, the objective is precisely the square of the volume of the parallelepiped spanned by the selected vectors. The problem of finding the largest volume parallelepiped in a collection of given vectors has been studied~\cite{Nikolov2015,Khachiyan1996,Summa2015} for over three decades. Another interesting application is the determinantal point processes~\cite{kulesza2012determinantal}, where a probability distribution over subsets of vectors is defined. The probability of selecting a subset is defined to be proportional to the squared volume of the parallelepiped defined by them. These distributions display nice properties of negative correlation. Finding sets with the largest probability mass is exactly the determinant maximization problem. We refer the reader to \cite{nikolov2019proportional} for applications in experimental design and to ~\cite{anari2016nash} for application to fair allocations.

The computational complexity of the determinant maximization depends crucially on the combinatorial set family which constrains the set of feasible collection of vectors. The simplest constraint being the cardinality constraint, wherein the number of vectors is fixed, has been the most widely studied variant. For this, a variety of methods including convex programming based methods~\cite{allen2017near,Summa2015,Nikolov15,SinghX18}, combinatorial methods - such as local search and greedy~\cite{Khachiyan1996,MadanSTU19,lau2021local} - as well as close relationship to graph sparsification~\cite{allen2017near} have been exploited to obtain efficient approximation algorithms with very good guarantees. Overall, these results give a very clear understanding of the computational complexity of the problem. 

The more general case when the combinatorial constraints are defined by a matroid constraint has recently received extensive focus ~\cite{NikolovS16,anari2016nash,anari2017generalization,AnariGV18,madan2020maximizing}. This is especially interesting since some of the applications are naturally modeled as matroid constraints, in particular, as partition constraints. Unfortunately, there is a big gap between estimation algorithms and approximation algorithms in this case! Indeed, one can approximately \emph{estimate} the value of an optimal solution with a good guarantee, however, {\em finding} such a solution is much more challenging, leading to an exponential loss in the approximation factor. For example, even for the special case of the partition matroid, there is an $e^d$-approximate estimation algorithm but the best known approximation algorithms return a solution with an approximation factor of $e^{O(d^2)}$, an exponential blow-up\footnote{Since we are computing the determinant of $d\times d$ matrices, the exponent of $d$ in the approximation factor is appropriate. Indeed, many works even consider the $d^{\text{th}}$-root of the determinant where the approximation factors are also the $d^{\text{th}}$-root of the above bounds.}. A fundamental reason for this gap is the reliance on the relationship between convex programming relaxations for the problem and the theory of stable polynomials and its generalization to strongly log-concave polynomials. Unfortunately, these methods are inherently non-algorithmic and do not give a simple way to obtain efficient algorithms with the same guarantees that match the estimation bounds.

\subsection{Our Results and Contributions}

In this work, we introduce new combinatorial methods for determinant maximization under a matroid constraint and give an $O(d^{O(d)})$-deterministic approximation algorithm. While previous works have used a convex programming approach and the theory of stable polynomials, our approach builds on the classical matroid intersection algorithm. Our first result focuses on the case when the rank of the matroid is exactly $d$, i.e., the output solution will contain precisely $d$ vectors.

\begin{theorem}\label{thm:main}
There is a polynomial time algorithm which, given a collection of vectors $v_1,\ldots, v_n\in \mathbb{R}^d$ and a matroid $\M=([n],\I)$ of rank $d$, returns a set $S\in \I$ such that
$$\det\left(\sum_{i\in S} v_i v_i^\top  \right)=\Omega\left(\frac{1}{d^{O(d)}} \right)\max_{S^*\in \I} \det\left(\sum_{i\in S^*} v_i v_i^\top  \right).$$
\end{theorem}

Our results improve the $e^{O(d^2)}$-approximation algorithm which relies on the $e^{O(d)}$-estimation algorithm~\cite{AnariGV18,anari2017generalization,madan2020maximizing}. Our algorithm builds on the matroid intersection algorithm and is an iterative algorithm that starts at any feasible solution and improves the objective in each step. To maintain feasibility in the matroid constraint, each step of the algorithm is an exchange of multiple elements as found by an alternating cycle of an appropriately defined exchange graph.

\textbf{Result for $r \leq d$.} We also generalize the result when the rank  $r$ of the matroid is at most $d$. Observe that the solution matrix $\sum_{i\in S} v_i v_i^\top $ is a $d\times d$ matrix of rank at most $r$ and, therefore, the appropriate objective to consider is the product of its largest $r$ eigenvalues, or equivalently, the elementary symmetric function of order $r$ of its eigenvalues. Let $sym_r(M)$ be the $r^{th}$ elementary symmetric function of the eigenvalues of the $d\times d$ matrix $M$. Thus, our objective is to maximize $sym_r\left(\sum_{i\in S} v_i v_i^\top  \right)$.

\begin{theorem}\label{thm:main2}
There is a polynomial-time algorithm which, given a collection of vectors $v_1,\ldots, v_n\in \mathbb{R}^d$ and a matroid $\M=([n],\I)$ of rank $r\leq d$, returns a set $S\in \I$ such that
$$sym_r\left(\sum_{i\in S} v_i v_i^\top  \right)=\Omega\left(\frac{1}{r^{O(r)}} \right)\max_{S^*\in \I} sym_r\left(\sum_{i\in S^*} v_i v_i^\top  \right).$$
\end{theorem}

This again improves the best bound of $e^{{O(r^2)}}$-approximation algorithm based on $e^{O(r)}$-approximate estimation algorithms. The proof of Theorem~\ref{thm:main2} is presented in Appendix~\ref{sec:rlessd}.

\textbf{Technical Overview.} For intuition, let $\vol(S)$ denote the volume of the parallelepiped spanned by the vectors in $S$. Then $\vol(S)^2=\det\left(\sum_{i\in S}v_i v_i^\top \right)$, for any $S\subseteq U$ with $|S|=d$, so we can think of $\vol(S)$ as an equivalent objective function. First, observe that the feasibility problem of checking whether there is a set $S\in \I$ such that $\vol(S)>0$ can be reduced to matroid intersection. Indeed, the feasibility problem is equivalent to checking if there is a common basis of the matroid $\M$ and the linear matroid defined by the vectors $\{v_1,\ldots, v_n\}$. Since we aim to maximize $\vol(S)$ over all independent sets $S$, a natural approach is to use the weighted matroid intersection algorithm. Unfortunately, our weights are not linear, i.e., $\vol(S)$ does not equal $\sum_{i\in S} w_i$ or log-linear $\prod_{i\in S} w_i$ for some weights $w$ on the vectors. Nonetheless, the matroid intersection algorithm forms the backbone of our approach.

\paragraph{Overview of Matroid Intersection.} Before we describe our algorithm, let us review a classical algorithm to find a maximum weight common basis of two matroids. Given $U=\{1,\ldots, n\}$, a weight function $w:U\rightarrow \RR$ and two matroids $\M_1=(U, \I_1)$ and $\M_2=(U,\I_2)$, the goal is to find a common basis $S$ of maximum weight $w(S):=\sum_{e\in S} w_e$. We assume that there exists a common basis of the two matroids. Consider the following simple algorithm that also introduces some of the basic ingredients necessary for our algorithm. The algorithm will take as an input a common basis $S$ and either certify that $S$ is a maximum weight common basis or return a new common basis $\widehat{S}$ of higher weight. To explain the algorithm, we recall the important concept of the {\em exchange graph}. Given the set $S$, we construct a directed bipartite graph $G(S)$ with bipartitions given by $U\setminus S$ and $S$. For any $u\in U\setminus S$ and $v\in S$, $G(S)$ contains an arc from $u$ to $v$ if $S-v+u$ is a basis in $\M_2$ and an arc from $v$ to $u$ if $S-v+u$ is a basis in $\M_1$.  For convenience, we use $S-v+u$ to refer to the set $(S\cup\{u\})\setminus \{v\}$. Moreover, give each vertex $u\in U\setminus S$ a weight $-w_u$ and each vertex $v\in S$ a weight of $w_v$. A nice fact from matroid theory is that $S$ is a maximum weight basis if and only if there is no negative weight cycle in this directed graph (Chapter 41, Theorem 41.5 ~\cite{schrijver2003combinatorial}). Moreover, if $C$ is a directed negative weight cycle with minimum hops\footnote{Hops here refers to the number of arcs of the cycle.}, then $S\Delta C$ forms a common basis of the two matroids whose weight is strictly larger than the weight of $S$. Thus, the algorithm finds a maximum weight basis by iteratively finding a negative weight cycle in such an exchange graph.

With the above algorithm as a guiding tool, we describe our algorithm. The two matroids are precisely the constraint matroid $\M$ and the linear matroid defined over the vectors. A first challenge is that the objective function $\vol(S)^2 = \det \left(\sum_{i\in S} v_i v_i^\top \right)$ is not linear. Thus it is not possible to define the vertex weights as was done above. But a natural function to work with instead is the function $\log \vol(S)$, which is known to be submodular. While we do not use submodularity explicitly, our algorithm takes linear approximations of this function at each iteration while searching for improvements as in the matroid intersection algorithm. We use the geometric relationship between $\vol$ and $\det$ closely. The first new ingredient in our algorithm is to introduce \emph{arc} weights rather than vertex weights in the exchange graph $G(S)$. Indeed for the forward arcs $(u,v)$ for $u\not \in S$ and $v\in S$ that correspond to the linear matroid, we introduce a weight of $-\log \frac{\vol(S-v+u)}{\vol(S)}$. We also introduce a weight of $0$ for the backward arcs, which correspond to the arcs for the constraint matroid $\M$. The crucial observation is the following interpretation of the weight $\log \frac{\vol(S-v+u)}{\vol(S)}$: write the vector $u\not\in S$ in the basis $S$, i.e. $u=\sum_{v\in S} a_{uv} v$, for some $a_{uv}\in \RR$ for each $v\in S$. Then $\frac{\vol(S-v+u)}{\vol(S)}=|a_{uv}|$ (See Lemma \ref{lem:weight_w0}). Such relationships between the ratio of volumes and coefficients in expressing the vectors in basis given by $S$ play an important role. 

Our first lemma shows that if the volume of the current solution is \emph{much} smaller than the optimal solution, then there must be a cycle such that the sum of weights of the arcs on the cycle is significantly negative.

\begin{lemma}[Determinant to Cycle]\label{lemma:det-to-cycle}
Let $S$ be a basis of $\M$ and $OPT$ be the basis of $\M$ maximizing $\vol(OPT)$. If $\vol(OPT)\geq e^{5d\log d}\cdot  \vol(S)$, then there exists a directed cycle $C$ of $2\ell$ hops for some $\ell>0$ in the exchange graph $G(S)$ such that
$$\prod_{(u,v)\in C, u\notin S, v\in S} |a_{uv}|\geq 2(\ell!)^3=: f(\ell).$$
\end{lemma}

We call such a cycle an $f$-violating cycle. Observe that such a cycle can be found as a negative weight $2\ell$-hop cycle when weights are updated to $w_{\ell}(u,v)= \frac{1}{\ell}\log f(\ell)- \log |a_{uv}|$ for a forward arc $(u,v)$  where $u\not\in S$ and $v\in S$. The lemma relies on the following observation.  Abusing notation slightly, let $T$ and $S$ be matrices whose columns are the vectors in $OPT$ and $S$, respectively. Writing each vector in $OPT$ in the basis given by $S$ we obtain $T=SA^\top $ for some matrix $A$. The condition in the lemma implies that $\det(A)\geq e^{5d\log d}$. Also observe that the weight of any $(u,v)$ where $u\in OPT$ and $v\in S$ is exactly $-\log |a_{uv}|$ where $a_{uv}$ is the $(u,v)^{\text{th}}$-entry in $A$. Combining these facts, we can show  there exists a cycle satisfying the conditions of the lemma.

The next step in the algorithm is to find an $f$-violating cycle $C$ and then update the solution to $T=S\Delta C$. Again, we relate the change in objective $\vol(T)$ to the coefficients. While the entries $|a_{uv}|$ of the cycle are large, the objective of the new solution $T$ depends not only on the weight of the edges of the cycle but the weight of all arcs between all vertices in $C\setminus S$ and $C\cap S$. Indeed, consider a square matrix $B$ with rows and columns indexed by $C\setminus S$ and $S\cap C$ respectively with entry $(u,v)$ as $a_{uv}$. Recall $a_{uv}$ is the coefficient of vector $v$ when $u$ is expressed in basis $S$. Then $\vol(T)=|\det(B)|\cdot \vol (S)$ (Lemma~\ref{lem:vol-determinant}).  Thus it remains to lower bound the determinant of $B$. The entries on the diagonal of the matrix $B$ exactly correspond to entries that define the weights of the forward arcs on the cycle $C$. Thus Lemma~\ref{lemma:det-to-cycle} implies that the product of the diagonal entries of $B$ is large. In the next lemma, we show that if the cycle $C$ is the \emph{minimum hop} $f$-violating cycle, then we can in fact lower bound the determinant of the matrix.

\begin{lemma} [Cycle to Determinant] \label{lem:exch}
If $C$ is a minimum hop $f$-violating cycle in the exchange graph $G(S)$, then $\vol(S \Delta C) \geq 2 \cdot \vol(S)$. Moreover, $S\Delta C$ is also a basis of $\M$.
\end{lemma}

This lemma crucially uses the fact that $C$ is a minimum hop $f$-violating cycle as in the case for matroid intersection algorithms. Indeed, off-diagonal entries of the matrix $B$ correspond to arcs that form chords of the cycle $C$. The minimality of $C$ allows us to show upper bounds on all the off-diagonal entries of the matrix $B$. A technical calculation then allows us to lower bound the determinant.

\subsection{Related Work}
\textbf{Determinant Maximization under Cardinality Constraints.} Determinant maximization problems under a cardinality constraint have been studied widely~\cite{Khachiyan1996,Summa2015,Nikolov15,SinghX18,allen2017near,MadanSTU19}. Currently, the best approximation algorithm for the case $r\leq d$ is an $e^r$-approximation due to Nikolov~\cite{Nikolov15} and for $r\geq d$, there is an $e^d$-approximation~\cite{SinghX18}. It turns out that the problem gets significantly easier when $r>>d$, and there is a $(1+\epsilon)^d$-approximation when $r\geq d+\frac{d}{\epsilon}$~\cite{allen2017near,MadanSTU19,lau2021local}. These results use local search methods and are closely related to the algorithm discussed in this paper, as the cycle improving algorithm will always find a $2$-cycle when the matroid is defined by the cardinality constraint.

\textbf{Determinant Maximization under Matroid Constraints.} As mentioned earlier, determinant maximization under a matroid constraint is considerably challenging and the bounds also depend on the rank $r$ of the constraint matroid. There are $e^{O(r)}$-estimation algorithms when $r\leq d$~\cite{NikolovS16,anari2017generalization,AnariLGV19} and a $\min\{e^{O(r)}, O\bigl(d^{O(d)}\bigr)\}$-estimation algorithm when $r\geq d$~\cite{madan2020maximizing}. The output of these algorithms is a random feasible set whose objective is at least $\min\{e^{O(r)}, O\bigl(d^{O(d)}\bigr)\}$ of the objective of a convex programming relaxation, in expectation. Since the approximation guarantees are exponential, it can happen that the output set has objective zero almost always. To convert them into deterministic algorithms (or randomized algorithms that work with high probability), additional loss in approximation factor is incurred.  These results imply an $e^{O(d^2)}$-approximation algorithm when $r\leq d$, and a $O\bigl(d^{O(d^3)}\bigr)$-approximation algorithm~\cite{madan2020maximizing} for $r\geq d$.  Approximation algorithms are also known where the approximation factor is exponential in the size of the ground set for special classes of matroids~\cite{Ebrahimi17}.

\textbf{Nash Social Welfare and its generalizations.} A special case of the determinant maximization problem is the Nash Social Welfare problem~\cite{moulin2}. In the Nash Social Welfare problem, we are given $m$ items and $d$ players and there is a valuation function $v_i:2^{[m]}\rightarrow \RR_+$ for each player $i\in [d]$ that specifies value obtained by a player when given a bundle of items. The goal is to find an assignment of items to players to maximize the \emph{geometric mean} of the valuations of each of the players. When the valuation functions are additive, the problem becomes a special case of the determinant maximization and this connection can be utilized to give an $e$-approximation algorithm~\cite{anari2016nash}. Other methods including rounding algorithms~\cite{cole2015approximating,cole2017convex} as well as primal-dual methods~\cite{barman2018finding} have been utilized to obtain improved bounds. The problem has been studied when the valuation function is more general~\cite{garg2018approximating,barman2018greedy,anari2018nash,garg2021approximating} and a constant-factor approximation is known when the valuation function is submodular~\cite{li2022constant}.

\textbf{Other Spectral Objectives.} While we focus on the determinant objective, the problem is also interesting when considering other spectral objectives including minimizing the trace or the maximum eigenvalue of the $\left(\sum_{i\in S} \left(v_i v_i^\top \right)\right)^{-1}$. These problems have been studied for the cardinality constraint~\cite{allen2017near,nikolov2018proportional}. For the case of partition matroid, the problem of maximizing the minimum eigenvalue is closely related to the Kadison-Singer problem~\cite{marcus2015interlacing}.

\section{Algorithm for Partition Matroid}

We first show the algorithm and the analysis for a partition matroid with rank $d$. This allows us to show the basic ideas without going into the details of matroid theory. The generalizations to general matroid are quite standard. We detail them in Section~\ref{sec:matroid}.

Consider a partition matroid $\mathcal{M}$ with $d$ partitions $\mathcal{P}_1,\ldots, \mathcal{P}_d$, where each $\mathcal{P}_i$ contains $n_{i}$ vectors $v_{i1},\ldots,v_{in_{i}} \in \mathbb{R}^d$. Our goal is to find a set $S$ which provides a good approximation to the objective
\begin{align*}
    \max \left\{ \det(\sum_{v \in S} v v^\top): |S| = d, |S\cap \mathcal{P}_i| = 1 \; \forall i \right\}\,.
\end{align*}

Let $OPT$ denote the optimal solution set.  
The following theorem is a specialization of Theorem~\ref{thm:main} to the case of partition matroid.

\begin{theorem}\label{theorem:partition}
Given a partition matroid $\mathcal{M}$ with $d$ parts, let $OPT$ be the optimal solution to the determinant maximization problem on $\mathcal{M}$. Then, there is a polynomial-time deterministic algorithm that outputs a feasible set $S \in \mathcal{M}$ such that
\[ \det\left(\sum_{i\in S} v_i v_i^\top  \right)\geq e^{-10d\log(d)}\cdot \det\left(\sum_{i\in OPT} v_i v_i^\top  \right).\]
\end{theorem}

\subsection{Algorithm}

We begin by formally defining the exchange graph, the different weight functions, and then the algorithm which helps establish Theorem~\ref{thm:main} for the case of partition matroids.

\begin{definition} [Exchange Graph] Formally,  for a subset of vectors $S = \{v_1, v_2, \ldots, v_d\}$ with $v_i \in \mathcal{P}_i$ for all $i$, we define the exchange graph of $S$, denoted by $G(S)$ as a bipartite graph, where the right-hand side consists of vectors in $S$, i.e., $R = \{v_1, v_2, \ldots, v_d\}$ and the left-hand side consists of all the vectors $L = \bigcup_{i=1}^d \mathcal{P}_i\backslash\{v_i\}$ (See Figure \ref{fig:exchange_graph}).  Each $v_i \in R$ has an edge to every $u \in \mathcal{P}_i\backslash\{v_i\}$, i.e., all the vectors in the same part as $v_i$.  The vertices on the left-hand side have forward edges to every vertex in $S$.
\end{definition}
\begin{figure}
    \centering
    \includegraphics[page=3,width=0.6\textwidth]{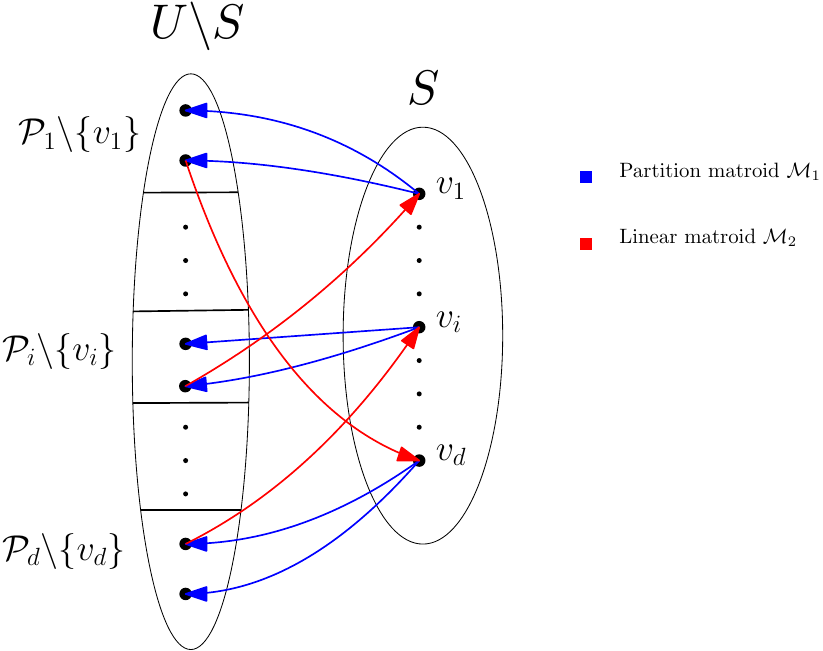}
    \caption{The exchange graph $G(S)$}
    \label{fig:exchange_graph}
\end{figure}

We define a family of weight functions on the exchange graph. The basic weight function will be denoted by $w_0:A(G(S))\rightarrow \RR$ and, in addition, we define weight functions $w_i$ for each $1\leq i\leq d$. To define these weights, we use the function $f:[d]\rightarrow \mathbb{Z}_+$ with $f(i) = 2(i!)^3$ for each $i>0$.
\begin{definition}[Weight functions on the Exchange graph]
We first define weight function $w_0$. All the backward arcs, from any $v_i \in S$ to every $u_j \in \mathcal{P}_i\backslash\{v_i\}$, have weight $0$. For $u_j \in L$, let $u_j =  \sum_{i=1}^d a_{ij} \cdot v_i$ be expression for $u_j$ in the basis $S$ where $a_{ij}\in \RR$ for each $i$.  Then the forward arc $(u_j,v_i)$ has weight $w_0(u_j,v_i) := - \log(|a_{ij}|)$ for each $i \in [d]$ and each $u_j\in L$.

Now we define the weight function $w_{\ell}$ on the arcs for any $1\leq \ell\leq d$. All backward arcs still have weight $0$ but every forward edge $(u, v)$ has weight $w_{\ell}(u, v) := \frac{\log(f(\ell))}{\ell} + w_0(u , v)$.
\end{definition}

The following lemma gives the intuition behind the weight function $w_0$ defined above. It shows that the weight on arc $(u_i,v_j)$ exactly measures the change in the objective when we replace element $v_i$ with $u_j$ in $S$. The proof appears in the appendix.

\begin{lemma}\label{lem:weight_w0}
Let $S$ be a solution with $\vol(S)>0$ and $u\not\in S$. Then for any $v\in S$, we have $w_0(u,v)=-\log \frac{\vol\left(S+u-v\right)}{\vol(S)}$.
\end{lemma}

While we will be specific about which weight function to use, but if it is not specified, then we refer to the weight function $w_0$.

\begin{definition}[Cycle Weight] The weight of a cycle $C$ in $G(S)$ is defined as $w_0(C) = \sum_{e \in C} w_0(e)$.
\end{definition}

Observe that the weight of a cycle depends only on the weight of the forward edges as backward edges have a weight $0$.

We want to move from the current set $S$ to a set with higher volume by exchanging on cycles in $G(S)$. But we want to exchange only on  cycles that satisfy certain nice properties. For this purpose, we define $f$-Violating Cycles and Minimal $f$-Violating Cycles. The algorithm will always exchange on a Minimal $f$-Violating Cycle.

\begin{definition}[$f$-Violating Cycle]
A cycle  in $G(S)$ is called an $f$-violating cycle if
\begin{equation*}
    w_0(C) < -\log f(|C|/2)\,,
\end{equation*}
where $|C|$ is the number of arcs in $C$.
\end{definition}

We have the following simple observation regarding $f$-violating cycle. 

\begin{observation}\label{obs:fcycle-weight}
	If $C$ is a $f$-violating cycle then $\prod_{(u, v)\in C: u\in L, v\in R} |a_{uv}| > 2\left(\left(\frac{|C|}{2}\right)!\right)^3.$
\end{observation}
(See \cref{sec:ommited_proof})

\begin{definition}[Minimal $f$-Violating Cycle]
A cycle $C$ in $G(S)$ is called a minimal $f$-violating cycle if
\begin{itemize}
    \item  $C$ is an $f$-violating cycle, and
    \item for all cycles $C'$ such that $V(C') \subset V(C)$, $C'$ is not an $f$-violating cycle.
\end{itemize}
\end{definition}

Note that finding an $f$-violating cycle with $2i$ arcs is equivalent to finding a negative cycle with $2i$ arcs in $G(S)$ with weights $w_i$. We use the following simple algorithm to find a minimal $f$-violating cycle in $G(S)$ (if one exists), where we iterate on the number of arcs in the cycle.
\begin{algorithm}[H]
\caption{Finding minimal $f$-violating cycle}
	\label{alg:minf}
\begin{algorithmic}
\For{$i = 1,\ldots, d$}
    \If{there is a negative cycle $C$ {with exactly $2i$ arcs} in $G(S)$ with weight function $w_i$}
       \State Return $C$
    \EndIf
\EndFor
\end{algorithmic}
\end{algorithm}

The following lemma is immediate. A proof appears in the appendix.
\begin{lemma}\label{lem:minf-alg}
Algorithm~\ref{alg:minf} finds the minimal $f$-violating cycle in $G(S)$.
\end{lemma}
After finding a minimal $f$-violating cycle,  $C$, we modify the current set $S$ to $S\Delta C$ and repeat. Observe that $S\Delta C$ is always a feasible set as it will pick exactly one element from each part. The main idea is that if $\vol(S)$ is small compared to $\vol(OPT)$, i.e., $\vol(S) < \vol(OPT)\cdot e^{-\Omega(d\log(d))}$, then there is always an $f$-violating cycle in $G(S)$ (see Lemma~\ref{thm:existence}). Moreover, if $C$ is a minimal $f$-violating cycle, then $\vol(S \Delta C) \geq 2 \cdot \vol(S)$ (see Lemma~\ref{thm:exch}). If we initialize $S$ to any solution with non-zero determinant, then the ratio $\vol(OPT)/\vol(S)$ is at most $2^{4\sigma}$ where $\sigma$ is the encoding length of our problem input (Chapter 3, Theorem 3.2 ~\cite{schrijver2000linear}). This implies that we need only modify the set $S$ polynomially many times before $\vol(S)$ becomes greater than $\vol(OPT)\cdot e^{-O(d\log(d))}$, which gives Theorem~\ref{theorem:partition}. Such an initialization can be obtained by finding a basis of $\RR^d$ that picks exactly one vector from each part. As discussed above, this problem can be solved by the matroid intersection algorithm over the partition matroid and the linear matroid defined by the vectors.

\begin{algorithm}[H]
\caption{Algorithm to find an approximation to $OPT$}
	\label{alg:exch}
\begin{algorithmic}
\State $S \leftarrow$ set with $|S| = d$, $|S\cap \mathcal{P}_i| = 1$ for all $i$, and $\vol(S) > 0$.
\While{there exists an $f$-violating cycle in $G(S)$}
    \State $C =$ minimal $f$-violating cycle in $G(S)$
    \State $S = S\Delta C$
\EndWhile
\State Return $S$

\end{algorithmic}
\end{algorithm}

\begin{lemma} \label{thm:existence}
For any set $S$ with $|S| = d$ and $\vol(S) > 0$, if $\vol(S) < \vol(OPT) \cdot e^{-5d\log(d)}$, then there exists an $f$-violating cycle in $G(S)$.
\end{lemma}
\begin{proof}
Let $OPT = \{u_1, u_2, \ldots, u_d\}$ and $S = \{v_1, v_2, \ldots, v_d\}$ such that $u_i, v_i \in \mathcal{P}_i$ for all $i \in [d]$. Observe that $(v_i,u_i)$ is an arc in the exchange graph for each $i$ since $u_i$ and $v_i$ belong to the same part\footnote{Given $u_{i}\neq v_{i}$}.

Abusing notation slightly, let $T$ and $S$ be matrices whose columns are the vectors in $OPT$ and $S$,  respectively.  Let $A$ be the coefficient matrix of $T$ w.r.t. $S$, i.e.,
$T = SA^\top $. Then
\begin{align*}
    \vol(OPT)^2 = \det(TT^\top) = \det(S A^\top A S^\top) = \det(SS^\top) \cdot |\det(A)|^2.
\end{align*}
Let $X =  OPT  \backslash S$, $Y = S \backslash OPT$, and $|X| = |Y| = k$. Without loss of generality, let $Y = \{v_1, \ldots, v_k\}$ and $X = \{u_1, \ldots, u_k\}$. Then $A = \begin{bmatrix} A_k & A'\\ 0 & I_{d-k} \end{bmatrix}$, where $A_k$ is the sub-matrix of $A$ corresponding to rows in $X$ and columns in $Y$. Then $\det(A) = \det(A_k)$.

As per the hypothesis in the lemma, we have $\det(SS^\top) < \det(TT^\top) \cdot e^{-10d\log(d)}$. Therefore,
\begin{equation}
    |\det(A_k)| > e^{5d\log(d)} \geq e^{5k\log(k)}. \label{eq:1}
\end{equation}
 By the Leibniz formula, we have $ \det(A_k) = \sum_{\sigma \in \mathcal{S}_k} \mathrm{sign}(\sigma)\prod_{i=1}^k a_{i\sigma(i)}$. Taking absolute values gives $|\det(A_k)| \leq \sum_{\sigma \in \mathcal{S}_k}\prod_{i=1}^k |a_{i\sigma(i)}|$. Since $|S_k| = k! \leq e^{k\log(k)}$, there exists a permutation $\sigma \in \mathcal{S}_k$ such that
  \begin{align}
\label{eq:0}  \prod_{i=1}^k |a_{i\sigma(i)}| > |\det(A_k)|\cdot e^{-k\log(k)} \geq e^{4k\log(k)}.
\end{align}

Let the cycle decomposition of this $\sigma$ be $\sigma = \{C_1, C_2, \ldots, C_\ell\}$. Then each $C_j$ corresponds to a unique cycle in $G(S)$ with $2|C_{j}|$ hops by considering the forward arcs $(u_i,v_{\sigma(i)})$ for each $i$ on the cycle and the backward arcs $(v_i,u_i)$ for each $i$ in $C_j$. We claim that at least one of these cycles is an $f$-violating cycle. If not, then by the definition of $f$-violating cycles, we have $\prod_{i\in C_j} |a_{i\sigma(i)}| \leq 2(|C_j|!)^3$. Multiplying over all cycles in $\sigma$ gives \begin{align*}
    \prod_{i=1}^k |a_{i\sigma(i)}| = \prod_{j=1}^\ell \prod_{i \in C_j} |a_{i\sigma(i)}| < \prod_{j=1}^\ell 2(|C_j|!)^3 < 2^k (k!)^3 < e^{4k\log(k)},
\end{align*}
where the second last inequality follows from $\sum_{j=1}^\ell |C_j| = k$. This contradicts~\cref{eq:0}, so $G(S)$ must contain an $f$-violating cycle. 
\end{proof}

The requirement in Lemma \ref{thm:existence} that $\vol(S) < \vol(OPT) \cdot e^{-5d\log(d)}$ is tight, up to the coefficient in the exponent. Consider the case where $d$ is a power of two (or more generally, any $d$ for which a Hadamard matrix of order $d$ is
known to exist), $S = \{e_1,\ldots, e_d\}$ consists of the standard basis vectors, and $L = H = \{h_1,\ldots, h_d\}$ consists of the columns of the $d\times d$ Hadamard matrix. The entries of $H$ are all $\pm1$, and $h_i^\top h_j = 0$ for $i\neq j$. Then $\vol(S) = 1$, and the optimal solution is $OPT =H$, which has objective value
\[ \vol(H) = \prod_{i=1}^d \|h_i\| = d^{d/2} = e^{\frac{d}{2}\log(d)} \cdot \vol(S),\]
since the vectors in $H$ are orthogonal. Meanwhile, the exchange matrix in this case is $A = H^\top $. Since all the entries of $A$ are $\pm1$, we know that the product of the entries along any cycle will have an absolute value of $1$. Thus, we cannot find an $f$-violating cycle in the same way, despite the fact that $\vol(S) \leq \vol(OPT) \cdot e^{-\frac{d}{2}\log(d)}$.

\subsection{Cycle Exchange and Determinant}
Now we show that exchanging on a minimal $f$-violating cycle $C$ increases the objective of the output set by at least a factor of two. The proof relies on two technical lemmas. First, observe that the arc weights given by $w_0(u,v)$ are exactly how much the objective will change if switch from the solution $S$ to $S+u-v\
$ in the solution. But switching on a cycle will switch multiple elements at the same time. Since our function $\vol(.)$ (or more appropriately $\log \vol(.)$) is not additive, it is not clear what the change in the objective. The following lemma characterizes exactly how the objective changes when we switch a large set.

 Consider our current solution $S$. Let $C$ be the minimal cycle found and $\ell = |C|/2$. Let $X = C\cap L$ and $Y = C\cap S$. Thus the output set $T=(S\cup X)\setminus Y$. We will also abuse notation to let $X, Y$ and $S$ represent the matrices whose columns are the vectors in their respective sets. Note that $S$ is $d\times d$ while both $X$ and $Y$ are $d\times \ell$. Observe that $\vol(S)^2=\det (SS^\top )$ and $\vol(T)^2=\det(TT^\top)=\det(SS^\top  + XX^\top  -YY^\top )$. Crucially, we show that the matrix consisting of coefficients $a_{uv}$ that define the weights on the arcs of the exchange graph for $u\in X$ and $v \in Y$ also defines the change in objective value.

\begin{lemma}\label{lem:vol-determinant}
	Let $S$ be a basis, let $X$ and $Y$ be sets with $|X| = |Y| = \ell$ and $Y\subseteq S$. Let $A$ be the $\ell \times d$ matrix of coefficients so that $X = SA^\top $, and let $A_C$ be the $\ell\times \ell$ submatrix of only the coefficients corresponding to columns in $Y$. If $T = (S \cup X)\del Y$ then $\vol(T)^2=\vol(S)^2\cdot \det(A_CA_C^\top )$.
\end{lemma}

Without loss of generality, let $C = (v_0 \rightarrow u_1 \rightarrow v_1 \rightarrow u_2 \rightarrow v_2\rightarrow \ldots  u_{\ell} \rightarrow v_{0})$ so that $X = \{u_1,\ldots u_\ell\}$ and $Y = \{v_1,\ldots, v_{\ell-1}, v_0\}$, and order the columns of $A_C$ accordingly so that the $\ell$-th column corresponds to $v_0$. Observe that diagonal entries of the $A_C$ correspond to coefficient of $v_i$ when expressing $u_i$ in basis of $S$ and thus equals $a_{ii}$. $C$ being $f$-violating implies that the product of the diagonal entries $\prod_{i=1}^\ell |a_{ii}| >f(\ell)$. To show that the volume of $T$ is large, we need to show $|\det(A_C)|$ is large. To this end, we utilize crucially that $C$ is the \emph{minimal} $f$-violating cycle.  Observe that the off-diagonal entries $a_{ij}$ exactly correspond to the weight on chords of the cycle. Since each chord introduces a cycle with smaller number of arcs, by minimality we know that it is not $f$-violating. This allows us to prove upper bounds on the off-diagonal entries of the matrix $A_C$. Finally, a careful argument allows us to give a lower bound on the determinant of any matrix with such bounds on the off-diagonal entries. We now expand on the above outline below.

\begin{lemma} \label{thm:exch}
If $C$ is a minimal $f$-violating cycle in $G(S)$, then $\vol(S \Delta C) \geq 2 \cdot \vol(S)$.
\end{lemma}
\begin{proof}
Let $C = (v_0 \rightarrow u_1 \rightarrow v_1 \rightarrow u_2 \rightarrow v_2\rightarrow \ldots  u_{\ell} \rightarrow v_{0})$ where $v_i, u_{i+1}$ belong to the same part and $v_i \in S$ (See Figure \ref{fig:exchange_cycle}).

\begin{figure}
    \centering
    \includegraphics[page=4,width=0.4\textwidth]{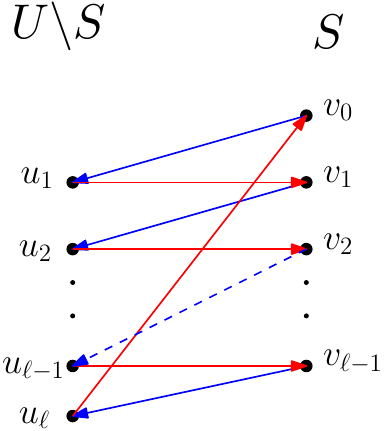}
    \caption{The cycle $C$}
    \label{fig:exchange_cycle}
\end{figure}

By the Lemma~\ref{lem:vol-determinant}, we know that $\vol(S \Delta C) = \det(A_C A_C^\top)^{1/2} \cdot \vol(S) = |\det(A_C)| \cdot \vol(S)$.  We will index the entries of $A_C$ according to the indices of $u_i$ and $v_j$ where the last column corresponds to $v_0$. Since $C$ has $2\ell$ hops, $A_C$ is an $\ell\times \ell$ matrix.

\begin{figure}
    \centering
    \begin{subfigure}{0.3\textwidth}
        \centering
        \includegraphics[page=5,width=\textwidth]{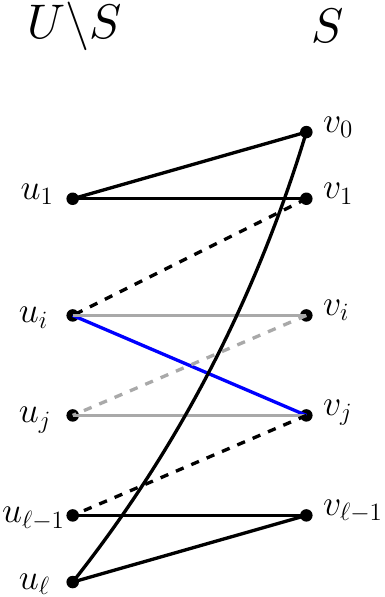}
        \caption{Shortcut when $i<j$}
    \end{subfigure}
    \hspace{1cm}
    \begin{subfigure}{0.3\textwidth}
        \centering
        \includegraphics[page=6,width=\textwidth]{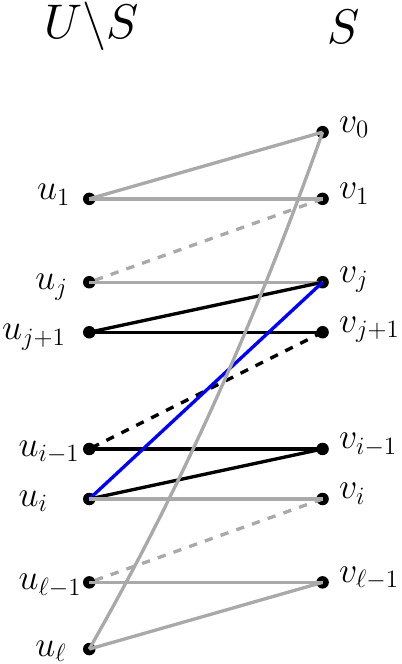}
        \caption{Sub-cycle when $i>j$}
    \end{subfigure}

    \caption{Structure when the edge $u_i \rightarrow v_j$ (in blue) is added}
    \label{fig:my_label}
\end{figure}

We now bound each entry of the matrix $A_C$ in terms of the its diagonal entries, $a_{i,i}$ for $i=1,\ldots, n$. We show upper bounds on the absolute value of each entry as a function of the diagonal entries. Consider the $i,j$-th entry of $A_C$. For $i > j$, define the cycle $C_{i,j} := (u_i \rightarrow v_j \rightarrow u_{j+1} \rightarrow v_{j+1}\rightarrow \ldots  v_{i-1} \rightarrow u_i)$. $C_{i,j}$ is a cycle with $2(i-j)$ hops and $V(C_{i,j}) \subset V(C)$. $C$ being a minimal $f$-violating cycle implies that $C_{i,j}$ is not an $f$-violating cycle. Therefore, $e^{-w_0(C_{i,j})} = |a_{i,j}| \cdot \prod_{s = j+1}^{i-1} |a_{s,s}| < f(i-j)$.
This implies
\begin{equation}
    |a_{i,j}| < \frac{f(i-j)}{\prod_{s = j+1}^{i-1} |a_{s,s}|}.
\end{equation}

Similarly for $j = \ell$, we have $ |a_{i,\ell}| < \frac{f(i)}{\prod_{s = 1}^{i-1} |a_{s,s}|}\,.$

For $i < j < \ell$, define $C'_{i,j} := (v_0 \rightarrow u_1 \rightarrow v_1\rightarrow \ldots u_i \rightarrow v_j\rightarrow \ldots  u_{\ell} \rightarrow v_{0})$. Again, $C'_{i,j}$ is a cycle with $2(\ell - j+i)$ hops which is not $f$-violating. Therefore, \begin{equation}
   |a_{i,j}| \cdot \prod_{s = 1}^{i-1} |a_{s,s}| \cdot \prod_{s = j+1}^{\ell} |a_{s,s}| < f(\ell - j + i). \label{eq:02}
\end{equation}
Since $C$ is an $f$-violating cycle, we also have \begin{equation}
    \prod_{s = 1}^\ell |a_{s,s}| > f(\ell). \label{eq:04}
\end{equation}
Combining~\eqref{eq:02} and~\eqref{eq:04} gives
 \begin{align*}
     |a_{i,j}| &< \frac{f(\ell-j+i)}{f(\ell)} \cdot \prod_{s=i}^j |a_{s,s}|.
\end{align*}
Let $B_{\ell}$ be the matrix obtained by applying the following operations to $A_C$ 
\begin{itemize}
    \item Multiply the last column by $a_{1,1}$ and for $j < \ell $, divide the $j$-th column by $\prod_{s=2}^{j} a_{s,s}$
    \item Divide the first row by $a_{1,1}$ and for $i > 1$, multiply the $i$-th row by $\prod_{s = 2}^{i-1} a_{s,s}$
    \item Divide the last column by $f(\ell)$ and, if needed, flip the sign of the last column so that $a_{\ell, \ell} > 0$.
\end{itemize}
Then $|\det(A_C)| = f(\ell) \cdot |\det(B_{\ell})|$, and $B_{\ell}$ satisfies the following properties:
\begin{itemize}
    \item $b_{i,i}$ = 1 for all $i \in [\ell-1]$, $b_{\ell, \ell} \geq 1$,
    \item $|b_{i,j}| \leq f(i-j)$ for all $j < i \leq \ell$, and
    \item $|b_{i,j}| \leq f(\ell - j + i)/f(\ell)$ for all $i < j \leq \ell$.
\end{itemize}

For $\ell \geq 2$, we have the following claim:
\begin{claim}\label{claim:mat_det}
  $\det(B_2) \geq 0.75$ and $\det(B_{\ell}) > 0.1$ for all $\ell \geq 3$.
\end{claim}
With this claim in hand, it implies that $|\det(A_C)| > 0.1 \cdot f(\ell) > 2$ for all $\ell \geq 3$. For $\ell = 2$, $\det(B_{2}) \geq 0.75$, and $\det(A_C) \geq 0.75 \cdot f(2) > 2$. Therefore, $\vol(S\Delta C) \geq |\det(A_C)| \cdot \vol(S) \geq 2 \cdot \vol(S)$.
\end{proof}

\begin{proof}[of Claim ~\ref{claim:mat_det}]
Consider the following process on $B_{\ell}$:
\begin{algorithm}[H]
\caption{Gaussian Elimination Process (Column Operations)}\label{alg:gaussian_elimination}
\begin{algorithmic}
\For{$s = 1, \ldots, \ell$} \Comment{Outer Loop}
    \For{$j = s+1, \ldots, \ell$} \Comment{Inner Loop}
    \State $b_{:,j} = b_{:,j} - b_{:,s} \cdot \frac{b_{s,j}}{b_{s,s}} $
    \EndFor
\EndFor
\end{algorithmic}
\end{algorithm}

Note that $\det(B_2) \geq 0.75$, $\det(B_3) \geq 0.73$, and $\det(B_4) \geq 0.83$ (see the end of the Appendix). From hereafter, we will assume that $\ell \geq 5$.

The output of the Algorithm~\ref{alg:gaussian_elimination} is a lower triangular matrix.
Let $b_{i,j}(s)$ denote the value of $b_{i,j}$ before the $s$-th iteration of the outer loop of Gaussian Elimination. For example, $b_{i,j}(1) = b_{i,j}$ for all $i,j$.

For any $i < j$, $b_{i,j}$ becomes $0$ at the end of the $i$-th iteration of the outer loop of the algorithm, and does not change after that. So, the final value of $b_{i,j}$, before it becomes $0$, is $b_{i,j}(i)$. Similarly, for $i \geq j$, the value of $b_{i,j}$ does not change after the $(j-1)$-th iteration of the outer loop, and therefore the final value of $b_{i,j}$, i.e., $b_{i,j}(\ell)$ satisfies $b_{i,j}(\ell) = b_{i,j}(j)$.

Since this process does not change the determinant of $B_\ell$, we have $\det(B_\ell) = \prod_{j=1}^\ell b_{j,j}(j)$. By Lemma \ref{lem:gaussian}, $b_{j,j}(j) > 1-0.92/\ell$ for $j < \ell$ and $b_{\ell, \ell}(\ell) > 0.303$.
Therefore,
\begin{align*}
  \det(B_{\ell}) = \prod_{j=1}^\ell b_{j,j}(j) \geq \left(1-\frac{0.92}{\ell}\right)^{\ell-1}\cdot 0.303\,.
\end{align*}
The function $\left(1-\frac{0.92}{\ell}\right)^{\ell-1}$ is a decreasing function of $\ell$, but has a horizontal asymptote at $\sim0.39$. Thus, $\left(1-\frac{0.92}{\ell}\right)^{\ell-1} \geq 0.39$ and this gives
\begin{equation*}
     \det(B_\ell) > 0.39 \times 0.303 > 0.1\,.
\end{equation*}
\end{proof}

\begin{lemma} \label{lem:gaussian}
For $\ell \geq 5$, the final values of entries of $B_\ell$ after Algorithm \ref{alg:gaussian_elimination} are bounded as follows:
\begin{enumerate}
    \item $|b_{i,j}(j)| < \binom{i}{j} \cdot f(i-j) $ for $1< j < i$, \label{itm:a}
    \item $|b_{i,j}(i)| < 1.5 \cdot f(\ell - j + i)/f(\ell)$ for $i < j < \ell$,\label{itm:b}
    \item $|b_{i,\ell}(i)| < 2.84 \cdot f(i)/f(\ell)$ for $i < \ell$,\label{itm:c}
    \item $ b_{j,j}(j) >  1- \frac{0.92}{\ell}$ for all $j < \ell$,\label{itm:d}
    \item $ b_{\ell,\ell}(\ell) > 0.303 $ .\label{itm:e}
\end{enumerate}
\end{lemma}
\begin{proof}
We will prove the lemma by induction on $j$, the column index. Note that Algorithm \ref{alg:gaussian_elimination} does not change the values of the first column of $B_\ell$, and it also does not change the values of the first row of $B_\ell$ before they become $0$. So, the bounds are trivially true for the first column and the first row.

For $i \geq j$,
\begin{equation}
    b_{i,j}(j) = b_{i,j}(1) - \sum_{s = 1}^{j-1} b_{i,s}(s) \cdot \frac{b_{s, j}(s)}{b_{s,s}(s)} \,. \label{eq:g1}
\end{equation}

Taking absolute values gives
\begin{equation}
    |b_{i,j}(j) - b_{i,j}(1)| \leq \sum_{s = 1}^{j-1} |b_{i,s}(s)| \cdot \frac{|b_{s, j}(s)|}{|b_{s,s}(s)|}\,. \label{eq:g2}
\end{equation}
The induction hypothesis implies that for  all $s < j$, $|b_{i,s}(s)| < \binom{i}{s} \cdot f(i-s)$, $|b_{s, j}(s)| < 1.5 \cdot f(\ell-j+s)/f(\ell)$, and $b_{s,s}(s) > 1-0.92/\ell \geq 0.816$ (since $\ell \geq 5$). Plugging these bounds in~\eqref{eq:g2}, we get
\begin{align}
    |b_{i,j}(j) - b_{i,j}(1)| &<  \frac{1.5}{0.816} \cdot \sum_{s = 1}^{j-1} \binom{i}{s} \cdot f(i-s) \cdot \frac{f(\ell-j+s)}{f(\ell)}\,. \label{eq:g3}
\end{align}
Note that
\begin{equation*}
    \frac{f(i-s)\cdot f(\ell-j+s)}{f(i-j)\cdot f(\ell)} = \frac{((i-s)!)^3 \cdot ((\ell-j+s)!)^3}{((i-j)!)^3 \cdot (\ell!)^3} = \left(\frac{\binom{\ell+i-j}{\ell}}{\binom{\ell+i-j}{i-s}}\right)^3 \,.
\end{equation*}
For any $1 \leq s \leq j-1$, $\frac{\binom{\ell+i-j}{\ell}}{\binom{\ell+i-j}{i-s}} \leq \frac{(i-j+1)}{\ell}$. Therefore,
\begin{align*}
    \frac{f(i-s)\cdot f(\ell-j+s)}{f(i-j)\cdot f(\ell)} &\leq  \frac{\binom{\ell+i-j}{\ell}}{\binom{\ell+i-j}{i-s}} \cdot \frac{(i-j+1)^2}{\ell^2}\,.
\end{align*}
Plugging this in~\eqref{eq:g3} gives
\begin{align}
    |b_{i,j}(j) - b_{i,j}(1)| &\leq f(i-j)\left( 1.84\cdot \frac{(i-j+1)^2}{\ell^2} \cdot \sum_{s = 1}^{j-1} \binom{i}{s} \cdot \frac{\binom{\ell+i-j}{\ell}}{\binom{\ell+i-j}{i-s}} \right) \notag\\
    &\leq f(i-j)\left( 1.84\cdot \frac{(i-j+1)^2}{\ell^2} \cdot\frac{i! \binom{\ell+i-j}{\ell}}{(\ell+i-j)!} \sum_{s = 1}^{j-1} \frac{(\ell-j+s)!}{s!}\right) \notag\\
    &= f(i-j)\left( 1.84\cdot \frac{(i-j+1)^2}{\ell^2} \cdot\frac{i! \binom{\ell+i-j}{\ell}(\ell-j)!}{(\ell+i-j)!} \sum_{s = 1}^{j-1} \binom{\ell-j+s}{\ell-j}\right) \notag\\
    &= f(i-j)\left( 1.84\cdot \frac{(i-j+1)^2}{\ell^2} \cdot\frac{i!(\ell-j)! }{\ell!(i-j)!} \sum_{s = 1}^{j-1} \binom{\ell-j+s}{\ell-j}\right)\,.\label{eq:g4}
\end{align}

For positive integers $a, b, x$ with $ x \leq a \leq b$,
\begin{equation}
  \binom{a}{x} + \binom{a+1}{x} + \binom{a+2}{x} + \ldots + \binom{b}{x} = \binom{b+1}{x+1} - \binom{a}{x-1}\,. \label{eq:g5}
\end{equation}
Using~\eqref{eq:g5} with $a = \ell-j+1$, $b = \ell-1$, and $x = \ell - j$ gives $\sum_{s = 1}^{j-1} \binom{\ell-j+s}{\ell-j} \leq \binom{\ell}{\ell-j+1}$ and from~\eqref{eq:g4},
\begin{align}
    |b_{i,j}(j) - b_{i,j}(1)| &\leq f(i-j)\left(  1.84\cdot \frac{(i-j+1)^2}{\ell^2} \cdot\frac{i!(\ell-j)!}{(i-j)!\ell!} \cdot \binom{\ell}{\ell-j+1}\right) \notag\\
     &= f(i-j)\left( 1.84\cdot \binom{i}{j} \cdot \frac{(i-j+1)^2j}{\ell^2(\ell-j+1)} \right) \label{eq:g6}\\
     &\leq f(i-j)\left( 1.84\cdot \binom{i}{j} \cdot \frac{(\ell-j+1)j}{\ell^2} \right).\label{eq:g7}
\end{align}

Since $(\ell-j+1)j$ is maximized at $j = (\ell+1)/2$, we have $ \frac{(\ell-j+1)j}{\ell^2}\leq \frac{(\ell+1)^2}{4\ell^2} \leq 0.36$ for any $\ell \geq 5$.

Plugging this in~\eqref{eq:g7} gives
\begin{align*}
    |b_{i,j}(j)| &\leq |b_{i,j}(j)| + f(i-j)\cdot 0.6624 \cdot \binom{i}{j}  \leq f(i-j)\left(1+ 0.6624\cdot \binom{i}{j} \right).
\end{align*}

Now we will restrict ourselves to the case when $i > j$. For $i = 2$, $j$ can only be $1$ and this corresponds to an entry in the first column for which the bounds are trivially true. So, we only need to consider $i \geq 3$.
Since $1 \leq j < i$, we have $\binom{i}{j} \geq i$. Furthermore, since $\ell \geq 5$ and $i \geq 3$, we have $1 \leq 0.3376\cdot i < 0.3376\binom{i}{j}$. This gives
\begin{align*}
    |b_{i,j}(j)| &\leq f(i-j)\left( 0.3376\cdot \binom{i}{j} +0.6624\cdot \binom{i}{j} \right)
    \leq  f(i-j) \cdot \binom{i}{j}.
\end{align*}
This concludes the proof of part \ref{itm:a}.

For $j > i$, we have
\begin{equation}
 \label{eq:5}    |b_{i,j}(i) - b_{i,j}(1)| \leq \sum_{s = 1}^{i-1} |b_{i,s}(s)| \cdot \frac{|b_{s, j}(s)|}{|b_{s,s}(s)|}\,.
\end{equation}

By the induction hypothesis, $|b_{s,j}(s)| < 1.5 \cdot f(\ell-j+s)/f(\ell)$, $|b_{i, s}(s)| < \binom{i}{s} \cdot f(i-s)$, and $b_{s,s}(s) > 1-0.92/\ell \geq 0.816$. Plugging these bounds in~\eqref{eq:5}, we get
\begin{align}
    |b_{i,j}(i)- b_{i,j}(1)| &<   \frac{1.5}{0.816}\sum_{s = 1}^{i-1}  \frac{f(\ell-j+s)}{f(\ell)} \cdot \binom{i}{s} \cdot f(i-s). \label{eq:temp}
\end{align}
Note that $\frac{f(\ell-j+s)\cdot f(i-s)}{f(\ell-j+i)} = \frac{2((\ell-j+s)!)^3 \cdot ((i-s)!)^3}{((\ell-j+i)!)^3 } = 2 \cdot \left(\frac{1}{\binom{\ell-j+i}{i-s}}\right)^3 $. For any $1 \leq s \leq i-1$, $\frac{1}{\binom{\ell-j+i}{i-s}} \leq \frac{1}{\ell-j+i}$. Therefore,
 \begin{align*}
    \frac{f(\ell-j+i-s)\cdot f(s)}{f(\ell-j+i)} \leq  \frac{1}{\binom{\ell-j+i}{i-s}} \cdot \frac{2}{(\ell-j+i)^2} \,.
\end{align*}
Plugging this in~\eqref{eq:temp} gives
\begin{align*}
    |b_{i,j}(i) - b_{i,j}(1)| &\leq  \frac{f(\ell-j+i)}{f(\ell)}  \cdot \frac{3.68}{(\ell-j+i)^2} \cdot \left(\sum_{s = 1}^{i-1}  \binom{i}{s} \cdot  \frac{1}{\binom{\ell-j+i}{i-s}} \right) \\
    &=  \frac{f(\ell-j+i)}{f(\ell)} \cdot  \frac{3.68\cdot i!(\ell-j)!}{(\ell-j+i)^2\cdot(\ell-j+i)!}\cdot \sum_{s = 1}^{i-1}   \binom{\ell-j+s}{\ell-j}\,.
\end{align*}
Using~\eqref{eq:g5} again, we get $\sum_{s = 1}^{i-1}   \binom{\ell-j+s}{\ell-j} \leq \binom{\ell-j+i}{\ell-j+1}$ and this gives
\begin{align}
     |b_{i,j}(i) - b_{i,j}(i)| &\leq  \frac{f(\ell-j+i)}{f(\ell)} \cdot \frac{3.68\cdot i!(\ell-j)!}{(\ell-j+i)^2\cdot(\ell-j+i)!} \cdot \binom{\ell-j+i}{\ell-j+1} \notag\\
      &=  \frac{f(\ell-j+i)}{f(\ell)} \cdot \frac{3.68\cdot i}{(\ell-j+i)^2\cdot(\ell-j+1)}\,. \label{eq:g8}
\end{align}
The function $\frac{i}{(\ell-j+i)^2}$ is maximized at $i = \ell-j$. So for any $j < \ell$, we have
\begin{align*}
     |b_{i,j}(i) - b_{i,j}(1)| &\leq   \frac{f(\ell-j+i)}{f(\ell)} \cdot \frac{3.68}{4(\ell-j)\cdot(\ell-j+1)} \leq 0.5 \cdot   \frac{f(\ell-j+i)}{f(\ell)}.
\end{align*}
Using the fact that $|b_{i,j}(1)| \leq f(\ell-j+i)/f(\ell) $, we have $|b_{i,j}(i)| \leq 1.5 \cdot f(\ell-j+i)/f(\ell)$ for $i < j < \ell$.

For $j = \ell$ and $i \geq 2$, equation~\eqref{eq:g8} gives
\begin{align*}
     |b_{i,\ell}(i)- b_{i,\ell}(1)| &\leq   \frac{f(i)}{f(\ell)} \cdot \frac{3.68}{i}  \leq 1.84 \cdot   \frac{f(i)}{f(\ell)},
\end{align*}
and therefore $|b_{i,\ell}(i)| \leq 2.84\cdot \frac{f(i)}{f(\ell)}$.
This concludes the proof of parts \ref{itm:b} and \ref{itm:c}.

For $i = j$ and $j < \ell$, using~\eqref{eq:g6}, we get
\begin{align*}
  |b_{j,j}(j) -1| \leq  \frac{1.84\cdot j}{\ell^2(\ell-j+1)} \leq  \frac{1.84\cdot (\ell-1)}{2\ell^2} \leq \frac{0.92}{\ell}.
\end{align*}

For $i = j = \ell$, by~\eqref{eq:g3} and the induction hypothesis,
\begin{align*}
    |b_{\ell,\ell}(\ell) - b_{\ell, \ell}(1)| &<  \frac{2.84}{1-0.92/\ell}\sum_{s = 1}^{\ell-1} \binom{\ell}{s} \cdot f(\ell-s) \cdot \frac{f(s)}{f(\ell)}\,.
\end{align*}
Following the proof outline of equation~\eqref{eq:g6} gives $|b_{\ell,\ell}(\ell)- b_{\ell, \ell}(1)| \leq   \frac{2.84}{0.816}\cdot  \frac{1}{\ell} \leq 0.697$.
Since $b_{\ell, \ell}(1) \geq 1$, we have $b_{\ell,\ell}(\ell) \geq b_{\ell,\ell}(1) - 0.697 \geq 0.303$.
\end{proof}

\section{Update Step for General Matroids}\label{sec:matroid}

Consider the case when $\mathcal{M} = ([n], \mathcal{I})$ is a general matroid of rank $d$. When we exchange on a cycle $C$ and update $S\leftarrow S\Delta C$, the resulting set is guaranteed to be independent in the linear matroid because of the determinant bounds in Lemma~\ref{thm:exch}, but it is not clear that it would be independent in the general constraint matroid, $\mathcal{M}$, when $\mathcal{M}$ is not a partition matroid. However, by exchanging on a minimal $f$-violating cycle in our algorithm, we can make the same guarantee.

In this section, we prove the existence of an $f$-violating cycle for any matroid $\mathcal{M}$ with rank $d$ when the current basis $S$ is sufficiently smaller in volume than the optimal solution $OPT$. We also prove that exchanging on a minimal $f$-violating cycle preserves independence in $\mathcal{M}$.

\begin{theorem}
For any basis $S$ with $|S| = d$ and $\vol(S) > 0$, if $\vol(S) < \vol(OPT) \cdot e^{-5d\log(d)}$, then there exists an $f$-violating cycle in $G(S)$.
\end{theorem}
\begin{proof}
Since $S$ and $OPT$ are independent and $|S| = |OPT|$, there exists a
perfect matching between $OPT\backslash S$ and $S \backslash OPT$ using the backward arcs in $G(S)$ (Chapter 39, Corollary 39.12a, ~\cite{schrijver2003combinatorial}). Let $X = OPT \backslash S$, $Y = S \backslash OPT$, and $|X| = |Y| = k$. Without loss of generality, let $Y = \{v_1, \ldots, v_k\}$ and $X = \{u_1, \ldots, u_k\}$ such that $(v_i \rightarrow u_i)$ is an arc in $G(S)$ for all $i \in [k]$. 

Let $T$ and $S$ be matrices whose columns are the vectors in $OPT$ and $S$,  respectively.  Let $A$ be the coefficient matrix of $T$ w.r.t. $S$, i.e.,
$T = SA^\top$. Then $A = \begin{bmatrix} A_k & A' \\ 0 & I_{d-k} \end{bmatrix}$, where $A_k$ is the sub-matrix of $A$ corresponding to rows in $X$ and columns in $Y$. Then by the same proof as in~\ref{thm:existence}, there exists a permutation $\sigma \in \mathcal{S}_k$ such that
\begin{align}
   \prod_{i=1}^k |a_{i\sigma(i)}| > |\det(A)|\cdot e^{-k\log(k)} \geq e^{4k\log(k)}\,. \label{eq:exist2}
\end{align}

Let the cycle decomposition of $\sigma$ be $\sigma = \{C_1, C_2, \ldots, C_\ell\}$ where $C_i = (i_1 \rightarrow i_2 \rightarrow \ldots i_{j} \rightarrow i_1)$. Since there is an edge from $v_{\sigma(j)}$ to $u_{\sigma(j)}$ for all $j$, every cyclic permutation $C_i$ corresponds to a cycle $(u_{i_1} \rightarrow v_{i_2} \rightarrow u_{i_2} \rightarrow v_{i_3} \ldots \rightarrow u_{i_j} \rightarrow v_{i_1} \rightarrow u_{i_1})$ in $G(S)$. We claim that at least one of these cycles is an $f$-violating cycle. If not, then by the definition of $f$-violating cycles, we have $\prod_{i\in C_j} |a_{i\sigma(i)}| \leq 2(|C_j|!)^3$ for all $j \leq \ell$. Multiplying over all the cycles in $\sigma$ gives \begin{align*}
    \prod_{i=1}^k |a_{i\sigma(i)}| = \prod_{j=1}^\ell \prod_{i \in C_j} |a_{i\sigma(i)}| \leq \prod_{j=1}^\ell 2 (|C_j|!)^3 < e^{4k\log(k)},
\end{align*}
where the last inequality follows from $\sum_{j=1}^\ell |C_j| = k$. This contradicts~\eqref{eq:exist2}, so $G(S)$ must contain an $f$-violating cycle.
\end{proof}

\begin{lemma} \label{lem:matroid-indep}
    If $C$ is a minimal $f$-violating cycle in $G(S)$, then $S\Delta C$ is independent in $\mathcal{M}$.
\end{lemma}
\begin{proof}
For clarity, let $V(C)$ denote the vertex set of $C$. Let $T := S \Delta V(C)$ and let $|C| = 2\ell$. Lets consider the graph $G(S)$ with weights $w_\ell$, and define $w_\ell(D) := \sum_{e\in D} w_\ell(e)$ for any cycle $D$. Since $C$ is an $f$-violating cycle, $w_\ell(C) = w_0(C) + \log(f(\ell)) < 0$.

Let the set of backward arcs in $C$ be $N_1$, and the set of forward arcs be $N_2$.
For the sake of contradiction, assume that $T \notin \mathcal{I}$. Then, there exists a matching $N_1'$ on $V(C)$ consisting of only backward arcs such that $N_1 \neq N_1'$ (Chapter 39, Theorem 39.13, ~\cite{schrijver2003combinatorial}). Let $A$ be a multiset of arcs consisting of all arcs in $N_2$ twice and all arcs $N_1$ and $N_1'$ (with arcs in $N_1 \cap N_1'$ appearing twice). Consider the directed
graph $D = (V(C), A)$. Since $N_1 \neq N_1'$, $D$ contains
a directed circuit $C_1$ with $V(C_1) \subsetneq V(C)$. Every vertex in $V(C)$ has exactly two in-edges and two out-edges in $A$. Therefore, $D$ is Eulerian, and we can decompose $A$ into directed circuits $C_1, \ldots, C_k$. Since only arcs in $N_2$ have non-zero weights, we have $\sum_{i=1}^k w_\ell(C_i) = 2w_\ell(C)$.

Because $V(C_1) \subsetneq V(C)$, at most one cycle $C_j$ can have $V(C_j) = V(C)$. If for some $j$, $V(C) = V(C_j)$, then $w_\ell(C_j) = w_\ell(C)$ as $C_j$ must contain every edge in $N_2$. So, $\sum_{i \neq j} w_\ell(C_i) = w_\ell(C) < 0$ and there exists a cycle $C_i$ such that $V(C_i) \subsetneq V(C)$ and $w_\ell(C_i) < 0$. Otherwise $V(C_j) \subsetneq V(C')$ for all $j$ and  $\sum_{i} w_\ell(C_i) = 2w_\ell(C) < 0$. Again, there exists a cycle $C_i$ such that $V(C_i) \subsetneq V(C)$ and $w_\ell(C_i) < 0$.

Let $C'$ be the directed cycle such that $V(C') \subsetneq V(C)$ and $w_\ell(C') \leq w_\ell(C) \leq 0$.  Define $y :=|C'|/2$. Thus $w_\ell(C') = y\cdot \log(f(\ell))/\ell + w_0(C') < 0$. Since $y < \ell$, $\log(f(y))/y \leq \log(f(\ell))/\ell$.  Therefore $w_0(C') \leq -y\cdot \log(f(\ell))/\ell  \leq -\log(f(y))$.
So $C'$ is an $f$-violating cycle with $V(C') \subset V(C)$, which contradicts the fact that $C$ is a minimal $f$-violating cycle.
\end{proof}

\section{Acknowledgements}

Adam Brown and Mohit Singh were supported in part by NSF CCF-2106444 and NSF CCF-1910423. Adam Brown was also supported by an ACO-ARC Fellowship. Aditi Laddha was supported in part by a Microsoft fellowship and NSF awards CCF-2007443 and CCF-2134105. Madhusudhan Pittu is supported in part by NSF awards CCF-1955785 and CCF-2006953. Prasad Tetali was supported in part by the NSF grant DMS-2151283.

\bibliographystyle{alpha}
\bibliography{references}

\appendix

\section{Omitted Proofs}
\label{sec:ommited_proof}

\begin{proof}[of Lemma \ref{lem:weight_w0}]
	Recall the statement of the Lemma: Let $S$ be a solution with $\vol(S)>0$ and $u\not\in S$. Then for any $v\in S$, we have $w_0(u,v)=-\log \frac{\vol(S+u-v)}{\vol(S)}$.\\
	 Let $S = \{v_1,\ldots, v_d\}$ so that $v = v_1$ and write $u = \sum_{i=1}^d a_i v_i$. We can also write $v = v^\perp + \sum_{i=2}^d b_i v_i$ where $v^\perp$ is orthogonal to $S\del \{v\}$. Then $u = a_1v^\perp + \sum_{i=2}^d (a_1b_i + a_i)v_i$. For $X \subseteq \mathbb{R}^d$ with $|X| = k \leq d$, let $\vol(X)$ denote the $k$-dimensional volume of the parallelepiped spanned by $X$. Then
	 \[ \vol(S) = \vol(S-v) \cdot \|v^\perp\|,\]
	 while
	 \[\vol(S+u-v) = \vol(S-v) \cdot |a_1|\|v^\perp\|,\]
	 since the change in volume from adding a single new vector is proportional to the length of the component of that vector which is orthogonal to our current set. Thus
	 \[-\log\frac{\vol(S+u-v)}{\vol(S)} = -\log\frac{\vol(S-v) \cdot |a_{uv}|\|v^\perp\|}{\vol(S-v) \cdot \|v^\perp\|} = -\log|a_{uv}|.\]
\end{proof}
\begin{proof}[of Observation \ref{obs:fcycle-weight}]
	Recall the statement of the Observation: If $C$ is a $f$-violating cycle then $\prod_{(u, v)\in C: u\in L, v\in R} |a_{uv}| > 2(|C|/2!)^3$.\\
	If $C$ is $f$-violating then $\ell(C) < -\log f(|C|/2)$, where $\ell(C)$ is the sum of the $w_0$ edge weights in $C$, and $f(|C|/2) = 2((|C|/2)!)^3$. Note that $|C\cap R| = |C\cap L| = |C|/2$, so $f(|C/2|) = 2(|C\cap R|!)^3$. By expanding $\ell(C)$ we see that
	\begin{align*}
		\ell(C)
		&= \sum_{(u, v)\in C: u\in L, v\in R} w_0(u,v)\\
		&= -\log\left(\prod_{(u, v)\in C: u\in L, v\in R} |a_{uv}|\right).
	\end{align*}
	Since $\ell(C) < -\log f(|C|/2)$, we can take the exponential to remove the logarithms and attain the desired inequality.
\end{proof}

\begin{proof}[of Lemma \ref{lem:minf-alg}]
	Recall the statement of the Lemma: Algorithm~\ref{alg:minf} finds the minimal $f$-violating cycle in $G(S)$, if one exists.\\
	In the $i$th iteration of Algorithm~\ref{alg:minf} we determine if there is a negative cycle in $G(S)$ with weights $w_i$ and $2i$ hops, as follows. For each vertex of $G(S)$, we start an instance of Bellman-Ford (See Chapter 8, Section 8.3, ~\cite{schrijver2003combinatorial}) with that vertex as the root, and proceed for $2i$ iterations. For source $u$, after $2i$ iterations, we check whether the distance from $u$ to $u$ is negative. If so, we have found a negative cycle with at most $2i$ hops. Note that for weights $w_i$, any negative cycle with at most $2i$ hops is an $f$-Violating cycle. Since the $(i-1)$-th iteration of the algorithm ensured that there are no $f$-Violating cycles with at most $2(i-1)$ hops, a negative cycle in the $i$-th iteration (if any) must have exactly $2i$ hops.
	
	Suppose there is an $f$-violating cycle $C$ in $G(S)$, so that $\ell = |C|/2$. Then, with weight $w_\ell$, the total weight of the cycle $C$ is
	\[ w_\ell(C) = \sum_{(u,v) \in C} w_\ell(u,v) = \sum_{(u,v) \in C} \log(f(\ell))/\ell + w_0(u,v) = \log f(\ell) + w_0(C). \]
	Since $C$ is $f$-violating we know that $\log f(\ell) < - w_0(C)$, so the above calculation shows that $C$ has negative total weight with weights $w_\ell$.  This guarantees that Algorithm \ref{alg:minf} will return an $f$-violating cycle whenever one exists.
	
	Now suppose that $C$ is the cycle returned by Algorithm $\ref{alg:minf}$ and we must show that $C$ is minimal $f$-violating. Let $C'$ be another cycle such that $V(C')\subset V(C)$. Then $C'$ has fewer hops than $C$, but it was not returned in iteration $|C'|/2$, so we know that $C'$ must not be $f$-violating. Thus $C$ is indeed minimal.
\end{proof}

\begin{proof}[of Lemma \ref{lem:vol-determinant}]
    Recall the statement of the Lemma \ref{lem:vol-determinant}: Let $S$ be a basis, let $X$ and $Y$ be sets with $|X| = |Y| = \ell$ and $Y\subseteq S$. Let $A$ be the $\ell \times d$ matrix of coefficients so that $X = SA^\top $, and let $A_C$ be the $\ell\times \ell$ submatrix of only the coefficients corresponding to columns in $Y$. If $T = (S \cup X)\del Y$ then $\vol(T)^2=\vol(S)^2\cdot \det(A_CA_C^\top )$.
    
    We will abuse notation slightly to let $S,X,Y$ also denote the matrices with columns from their respective sets. Order the columns of $S$ so that $Y$ makes up the first $\ell$ columns of $S$. Let $A'$ be the $\ell\times (d-\ell)$ submatrix of $A$ consisting of the remaining columns not already in $A_C$. Then
    \[ T = S \begin{bmatrix} A_C & A' \\ 0 & I_{d-\ell} \end{bmatrix}^\top, \]
    which implies that
    \[ \det(T) = \det(S) \cdot \det(A_C). \]
\end{proof}

\textbf{Bounds on $\det(B_2)$, $\det(B_3)$, $\det(B_4)$ mentioned in Proof of Claim ~\ref{claim:mat_det}}: 

For $\ell = 2$, $b_{1,1}(1) = 1$,  $|b_{1,2}(1)| \leq 0.125$,  $|b_{2,1}(1)| \leq 2$,  and $|b_{2,2}(2)| \geq 0.75$.
\begin{equation*}
    \det(B_2) \geq \prod_{i=1}^2 b_{i,i}(i) \geq 0.75.
\end{equation*}

The bounds on final values of $B_3$ are:
\bgroup
\def\arraystretch{1.5}%
\begin{table}[H]
\centering
\begin{tabular}{ |c|c|c| }
\hline
$b_{1,1}(1) =  1$ & $ |b_{1,2}(1) | \leq 0.0371$ &  $|b_{1,3}(1) | \leq 0.00463$  \\
\hline
$ |b_{2,1}(1) | \leq 2.0$ & $b_{2,2}(2) \geq 0.92$ & $|b_{2,3}(2) | \leq 0.0463$ \\
\hline
$ |b_{3,1}(1) | \leq 16.0$ & $|b_{3,2}(2)| \leq 2.5926$  & $ b_{3,3}(3)  \geq 0.79$ \\
\hline
\end{tabular}
\begin{equation*}
    \det(B_3) \geq \prod_{i=1}^3 b_{i,i}(i) \geq 0.73.
\end{equation*}
\end{table}

The bounds on final values of $B_4$ are:

\begin{table}[h]
\centering
\begin{tabular}{ |c|c|c|c| }
\hline
$b_{1,1}(1) =  1$ & $ |b_{1,2}(1) | \leq 0.015625 $ &  $|b_{1,3}(1) | \leq 0.00057871$ &  $|b_{1,4}(1) | \leq 0.000073$ \\
\hline
$ |b_{2,1}(1) | \leq 2.0$ & $b_{2,2}(2) \geq 0.96875$ & $|b_{2,3}(2) | \leq 0.01678241$ & $|b_{2,4}(2) | \leq 0.00072338$\\
\hline
$ |b_{3,1}(1) | \leq 16.0$ & $|b_{3,2}(2)| \leq 2.25$  & $ b_{3,3}(3)  \geq 0.95$ & $| b_{3,4}(3)|  \leq 0.018463$ \\
\hline
$ |b_{4,1}(1) | \leq 432.0$ & $|b_{4,2}(2)| \leq 22.75$  & $ |b_{4,3}(3)|  \leq 2.645$ & $ b_{4,4}(4)  \geq 0.9$ \\
\hline
\end{tabular}
\end{table}
\begin{equation*}
    \det(B_4) \geq \prod_{i=1}^4 b_{i,i}(i) \geq 0.83.
\end{equation*}

\section{Rank $r \leq d$}
\label{sec:rlessd}
In this section, we prove Theorem \ref{thm:main2}. Consider a matroid $\mathcal{M} = ([n], \mathcal{I})$ with rank $r\leq d$. Starting with a basis $S$ with non-zero volume, we will use a slight modification of Algorithm \ref{alg:exch} to iteratively find a basis with strictly larger volume. However since the set $S$ is not full dimensional, our edge weight functions will be different. 

Let $S = \{v_1, v_2, \ldots, v_r\}$ be a basis of $\mathcal{M}$ with $\vol(S) > 0$. We can write any vector $u_i$ in $\mathcal{M}$ as
 \begin{equation*}
 u_i = \sum_{j=1}^r a_{i,j} v_j + u_i^\perp,
 \end{equation*}
 where $u_i^\perp$ is orthogonal to $\Span(S)$.

The change in volume on replacing some $v \in S$ by $u \notin S$ is given by 
\begin{equation}
    \frac{\vol(S - u + v)}{\vol(S)} = \sqrt{a_{uv}^2 + \frac{\norm{u^\perp}^2}{\norm{v^\perp}^2}},  \label{eq:cv}
\end{equation}
where $v^\perp$ is the component of $v$ orthogonal to $\Span(S-v)$. The two terms in equation~\eqref{eq:cv} have geometric meanings. Let us decompose $u$ into $u^\perp + u^\parallel$, where $u^\perp$ is the component of $u$ orthogonal to $\Span(S)$. Then $|a_{uv}|$ is exactly the change in the volume if we project $u$ to $\Span(S)$ before replacing $v$, i.e.,  $|a_{uv}| = \frac{\vol(S-v+u^\parallel)}{\vol(S)}$, and $\frac{\norm{u^\perp}}{\norm{v^\perp}}$ is the change in the volume if we project $u$ orthogonal to $\Span(S)$ before replacing $v$, i.e., $\frac{\norm{u^\perp}}{\norm{v^\perp}} = \frac{\vol(S-v+u^\perp)}{\vol(S)}$. So, we augment the exchange graph to reflect this.

Like Lemma \ref{thm:existence}, when $\sym_r(S) < \sym_r(OPT) \cdot r^{-\Omega(r)}$, we can find an $\tilde{f}$-violating cycle (for an appropriate function $\tilde{f}$) in the augmented exchange graph.

However unlike Lemma \ref{thm:exch}, the change in the objective induced by a cycle $C$ in the augmented graph is not a simple function of the weights of chords and arcs of $C$. To get around this issue, we use the geometric relation between $\sym_r$ and $\vol$, specifically the subadditivity of $\vol$ to relate $\sym_r$ to the chord and arc weights of $C$.

We define the exchange graph $G(S)$ exactly as Definition 1 with $w_0(u_i, v_j) = -\log(|a_{i,j}|)$. Our approach to find a basis with larger volume is to first try to exchange on an $f$-violating cycle in $G(S)$. Like  Algorithm 1, exchanging on an $f$-violating cycle implies increase in volume. Unlike Algorithm 1, failure to find an $f$-violating cycle does not imply that the volume of the current solution is close to optimal. So, we move to Stage 2, where we work with an augmented version of the exchange graph defined below.

We decompose every vector $u_i$ in $\mathcal{M}$ as
 \begin{equation*}
 u_i = u_i^\parallel + u_i^\perp,
 \end{equation*}
where $u_i^\parallel \in \Span(S)$ and $u_i^\perp$ is orthogonal to $\Span(S)$.

In the augmented exchange graph $\widetilde{G}(S)$, for every vector $u_i \in \mathcal{M}$, we create two vertices $u_i^\parallel$ (called a parallel vertex) and  $u_i^\perp$ (called a perpendicular vertex) in the left-hand side.

\begin{definition} [Augmented Exchange Graph] For a subset of vectors $S = \{v_1, v_2, \ldots, v_r\}$, we define the augmented exchange graph of $S$, denoted by $\widetilde{G}(S)$ as a bipartite graph, where the right-hand side consists of vectors in $S$, i.e., $R = \{v_1, v_2, \ldots, v_r\}$ and the left-hand side consists of all the vectors $L = \bigcup_{u \in U\del S } \{u_i^\perp , u_i^\parallel\}$.  For each $v_i \in R$, if $S - v_i + u \in \mathcal{I}$, then $v_i$ has an edge to $u^\perp$ and an edge to $u^\parallel$.  The vertices on the left-hand side have forward edges to every vertex in $S$ ( See figure \Cref{fig:augmented_graph}).
	\begin{figure}
		\centering
		\includegraphics[page=7,width=0.6\textwidth]{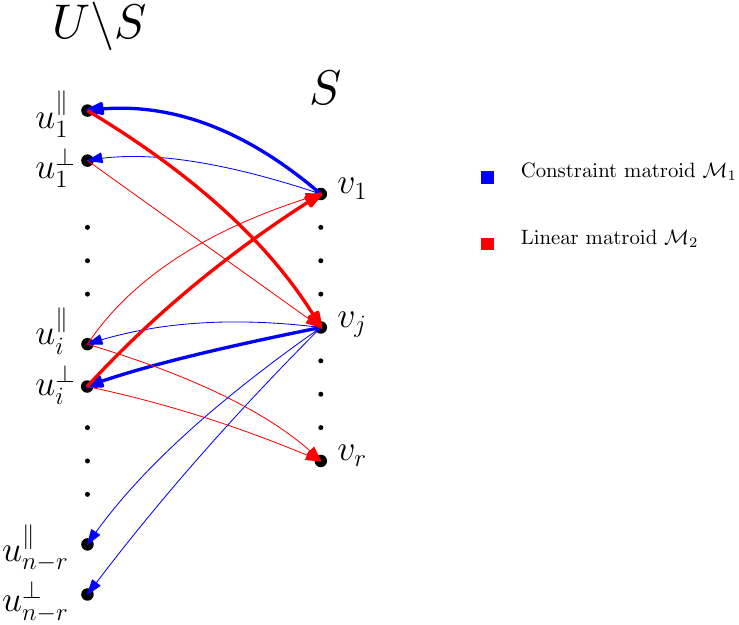}
		\caption{The augmented exchange graph $\widetilde{G}(S)$(An exchange cycle is shown in bold edges)}
		\label{fig:augmented_graph}
	\end{figure}
	
\end{definition}

Each vector $v_j \in S$ can be decomposed as $v_j = \sum_{i \neq j} \alpha_{j,i} v_i + v_j^\perp $, where $v_j^\perp $ is orthogonal to the span of $S\backslash \{v_j\}$. We will call $v_j^\perp $ the orthogonal component of $v_j$ and use it to define edge weights in $\widetilde{G}(S)$.

\begin{definition}[Weight functions on the Augmented Exchange graph]
All the arcs from some vertex in $R$ to some vertex in $L$, have weight $0$. The weights of the forward arcs are defined as
\begin{itemize}
\item for every parallel vertex $u_i^\parallel \in L$, the weight of the arc $u_i^\parallel \rightarrow v_j$ is 
\[\widetilde{w}_0(u_i^\parallel \rightarrow v_j) := -\log(|a_{i, j}|),\]
\item for every perpendicular vertex $u_i^\perp \in L$, the weight of the arc $u_i^\perp \rightarrow v_j$ is 
\[\widetilde{w}_0(u_i^\perp \rightarrow v_j) := -\log\left(\frac{\norm{u_i^\perp}}{\norm{v_j^\perp}}\right).\]
\end{itemize}

We define a family of weight functions on the exchange graph. To define these weights we use the new function $\Tilde{f}(i)=(i!)^{11}$ if $i\geq 2$ and $\tilde{f}(1)=2$. Now we define the weight function $\widetilde{w}_{\ell}$ analogously to $w_{\ell}$, i.e., all backward arcs still have weight $0$ but every forward edge $(u, v)$ has weight 
\[\widetilde{w}_{\ell}(u, v) := \frac{\log(\tilde{f}(\ell))}{\ell} + \widetilde{w}_0(u , v).\]
\end{definition}

\begin{observation}\label{obs:parallel-matroid}
    For a current solution $S$, let $\widetilde{\mathcal{M}}(S)$ be the matroid obtained from $\mathcal{M}$ by adding an element parallel to every $u \in U\del S$, and labelling the pair $u^\perp$, $u^\parallel$. Then the subgraph of $\widetilde{G}(S)$ induced by edges with finite weight is the matroid exchange graph where $\mathcal{M}_1 = \widetilde{\mathcal{M}}(S)$, and $\mathcal{M}_2$ is the linear matroid on $S \cup \bigcup_{u \in U\backslash S} \{u^\perp , u^\parallel\} $.
    
    By construction, no independent set in $\widetilde{\mathcal{M}}(S)$ contains both $u^\perp$ and $u^\parallel$, for any $u \in U\del S$. Thus, if $\widetilde{\mathcal{I}}$ is an independent set in $\widetilde{\mathcal{M}}(S)$ and $\mathcal{I}$ is obtained from $\widetilde{\mathcal{I}}$ by replacing each instance of $u^\perp$, or $u^\parallel$ with the original element $u \in U\backslash S$, then $\mathcal{I}$ is independent in $\mathcal{M}$.
\end{observation}

Similar to Lemma \ref{lem:weight_w0}, a function of $\widetilde{w}_0(u_i^\parallel,v_j)$ and $\widetilde{w}_0(u_i^\perp,v_j)$ measures the change in objective when we replace element $v_j$ by $u_i$.

\begin{lemma}\label{lem:weight_wa0}
Let $S$ be a solution with $\vol(S)>0$ and $u\not\in S$ with $u = u^\parallel + u^\perp$. Then for any $v\in S$, we have $\sqrt{e^{-2\widetilde{w}_0(u^\parallel,v)}+ e^{-2\widetilde{w}_0(u^\perp,v)}}=\frac{\vol\left(S - u + v\right)}{\vol(S)}$.
\end{lemma}

Our algorithm works in 2 stages. In the first stage, we try to find an $f$-violating cycle in the exchange graph $G(S)$. If the algorithm finds such a cycle, it exchanges on it. If no $f$-violating cycle is found, we move to the augmented exchange graph $\widetilde{G}(S)$, and search for an $\Tilde{f}$-violating cycle in $\widetilde{G}(S)$. If no such cycle is found in Stage 2, then Lemma \ref{thm:existence_extended} guarantees that $\vol(S) \geq e^{-O(r\log(r))} \cdot \vol(OPT)$.

\newcommand{\per}{\mathrm{per}}
\newcommand{\proj}{\mathbf{P}}
\begin{algorithm}[H]
\caption{Algorithm to find an approximation to $OPT$}
	\label{alg:exchr}
\begin{algorithmic}
\State $S \leftarrow$ basis with $\vol(S) > 0$.

\State Let $f(i) = 2(i!)^3$ and $\Tilde{f}(i) = (i!)^{11}$
\While{There exists an $f$-violating cycle in $G(S)$ or an $\Tilde{f}$-violating cycle in $\widetilde{G}(S)$}
\If{There exists an $f$-violating cycle in $G(S)$}
\State \textbf{Stage 1}:
    \State $C =$ minimal $f$-violating cycle in $G(S)$
    \State $S = S\Delta C$
\Else
\State \textbf{Stage 2}:
    \State $C =$ minimal $\Tilde{f}$-violating cycle in $\widetilde{G}(S)$
    \State $\widetilde{C} = \{ u \in [n]\del S: u^\perp \text{ or } u^\parallel \in C\} \cup \{ v \in S: v \in C\}$
    \State $S = S\Delta \widetilde{C}$
\EndIf
\EndWhile

\State Return $S$

\end{algorithmic}
\end{algorithm}
\begin{lemma}
If Algorithm \ref{alg:exchr} finds an $f$-violating cycle, $C$, in $G(S)$, then $\vol(S \Delta C) \geq 2 \cdot \vol(S)$.
\end{lemma}
\begin{proof}
Let $C'$ be the projection of $C$ onto $\Span(S)$. By Lemma~\ref{vol-determinant2}, we know that $\vol(S\Delta C)\geq \vol(S\Delta C')$, and by Lemma~\ref{thm:exch}, we know that $\vol(S\Delta C')\geq 2 \cdot \vol(S)$. Therefore,
\begin{equation*}
    \vol(S\Delta C)\geq \vol(S\Delta C')\geq 2 \cdot \vol(S),
\end{equation*}
which concludes the proof of the Lemma.
\end{proof}

From hereafter we analyze the case when Algorithm \ref{alg:exchr} does not find an $f$-violating cycle in Stage 1, moves on to Stage 2 and finds an $\Tilde{f}$-violating cycle in $\widetilde{G}(S)$.

\subsection{Existence of $\tilde{f}$-violating cycle}
To guarantee we make progress at every iteration, we need to ensure there will always be an $\tilde{f}$-violating cycle, whenever our current volume is far from the optimal volume.
Before we prove the existence of an $\tilde{f}$-violating cycle, we state a couple of useful observations.

For convenience, to specify the volume of a set of vectors $\{v_1, v_2, \ldots, v_r\}$, instead of writing $\vol(\{v_1, v_2, \ldots, v_r\})$, we use $\vol(v_1, v_2, \ldots, v_r)$.
\begin{observation}\label{obs:volume_triangle} For any set of vectors $\{v_{1}^{a},v_{1}^{b},v_{2},\dots,v_{r}\}$,
		\begin{align*}
			\vol(v_{1}^{a}+v_{1}^{b},v_{2},\dots,v_{r}) &\leq \vol(v_{1}^{a},v_{2},\dots,v_{r})+\vol(v_{1}^b,v_{2},\dots,v_{r})\,, \\
			\vol(v_{1}^{a}+v_{1}^{b},v_{2},\dots,v_{r}) &\geq \vol(v_{1}^{a},v_{2},\dots,v_{r})-\vol(v_{1}^b,v_{2},\dots,v_{r})\,.
			\end{align*}
		\end{observation}
	\begin{proof}
	Let $P^{\perp}$ be the projection matrix  orthogonal to $\Span(v_{2},\dots,v_{r})$. By triangle inequality, we have that
	\begin{align*}
		    \norm{P^{\perp}(v_{1}^{a}+v_{1}^{b})} &\leq \norm{P^{\perp}v_{1}^{a}}+\norm{P^{\perp}v_{1}^b}, \text{ and }\\
		    \norm{P^{\perp}(v_{1}^{a}+v_{1}^{b})} &\geq \norm{P^{\perp}v_{1}^{a}}-\norm{P^{\perp}v_{1}^b}.
	\end{align*}
Since $\vol(v_{1}^{a},v_{2},\dots,v_{r}) = \norm{P^{\perp}v_{1}^{a}} \cdot \vol(v_{2},\dots,v_{r})$, multiplying both sides by $\vol(v_{2},\dots,v_{r})$ gives us the required inequalities.
		\end{proof}
\begin{observation}\label{obs:log_vol_submodular}
	If $v_{i}^{\perp}$ is the orthogonal projection of $v_{i}$ onto $\Span(S-v_{i})$, then 
	\begin{align*}
    \vol(v_{1},\dots,v_{k}) \cdot \prod_{i=k+1}^{r}\norm{v_{i}^{\perp}} \leq \vol(v_{1},\dots,v_{r})
	\end{align*}
 for any $k\in [r]$.
\end{observation}
   \begin{proof}
   	 Let $S_{i}=\{v_{1},\dots,v_{i}\}$ for $i\in [r]$ and $S_{0}=\emptyset$. Since $\vol(v_{1},\dots,v_{r}) = \prod_{i=1}^{r}\frac{\vol(S_{i})}{\vol(S_{i-1})}$ where $vol(\emptyset)=1$, it suffices to prove that $\frac{\vol(S_{i})}{\vol(S_{i-1})}\geq \norm{v_{i}^{\perp}}=\frac{\vol(S)}{\vol(S\backslash \{v_{i}\})}$ which follows from the submodularity of $\log \vol(\cdot)$ as $S_{i-1}\subseteq S\backslash \{v_{i}\}$.
   	\end{proof}

The following lemma is an extension to \Cref{thm:existence} when the current solution has $r\leq d$ vectors.
\begin{lemma}\label{thm:existence_extended}
For any basis $S \in \mathcal{I}$ with $\vol(S) > 0$, if $\vol(OPT) > \vol(S) \cdot r^{2r} \cdot \Tilde{f}(r)$, then there exists an $\Tilde{f}$-violating cycle in $\widetilde{G}(S)$.
\end{lemma}
\begin{proof}
Since $S$ and $OPT$ are independent and $|S| = |OPT|$, there exists a
perfect matching between $OPT\backslash S$ and $S \backslash OPT$ using the backward arcs in $G(S)$ (Chapter 39, Corollary 39.12a, ~\cite{schrijver2003combinatorial}). Let $X = OPT \backslash S$, $Y = S \backslash OPT$, and $|X| = |Y| = \ell$. Without loss of generality, let $Y = \{v_1, \ldots, v_\ell\}$ and $X =\{u_1, \ldots, u_\ell\} $ such that $(v_i \rightarrow u_i)$ is an arc in $G(S)$ for all $i \in [\ell]$. Let $Z = OPT \cap S = \{v_{\ell+1}, \ldots, v_r\}$ and let us use $T$ instead of $OPT$ for ease of notation. 
 
From the hypothesis of the lemma, we have
\begin{align}
	r^{2r} \cdot \tilde{f}(r) \leq \frac{\vol(T)}{\vol(S)}&= \frac{\vol(X \cup Z)}{\vol(Y \cup Z)}.
\end{align}
Since $u_i = \sum_{j=1}^r a_{i,j} v_j + u_i^\perp$, we can decompose each vector $u_i$ into a sum of $r+1$ vectors, i.e., $ u_i = \sum_{j=0}^r u_{i}^{(j)}$
where
\begin{itemize}
    \item $u_{i}^{(0)} := u_i^\perp$ for all $i$, and
    \item $u_{i}^{(j)} := a_{i, j} v_j\,$. 
\end{itemize}

Now using Observation \ref{obs:volume_triangle}, we can expand $\vol(T)$ as
\begin{align}
	\frac{\vol(T)}{\vol(S)}&\leq \sum_{i_1, i_2, \ldots, i_{\ell} \in \{0,\ldots, r\}} \frac{\vol(\{u_{1}^{(i_1)}, \, u_{2}^{(i_2)}, \ldots, \,u_{\ell}^{(i_\ell)}\} \cup Z)}{\vol(S)}. 
\end{align}
If $i_j > \ell$ for any $j$, then $u_j^{(i_j)}$ and the vectors in set $Z$ are linearly dependent and therefore the volume is $0$. So we can restrict ourselves to the case when $i_j \in \{0, \ldots, \ell\}$ for all $j \in [\ell]$.

For any permutation $\sigma \in \mathcal{S}_\ell$, define $\mathcal{F}(\sigma) := \{(i_1, i_2, \ldots, i_\ell): i_j \in \{\sigma(j) , 0\} \}$.
Since $u_{i_1}^{(j)}$ and $u_{i_2}^{(j)}$ are linearly dependent whenever $j > 0$, we can rewrite $\vol(T)$ as
\begin{align}
    \vol(T) &\leq \sum_{\tau \in \bigcup_{\sigma} \mathcal{F}(\sigma)} \vol(\{u_{1}^{(\tau_1)}, \, u_{2}^{(\tau_2)}, \ldots, \,u_{\ell}^{(\tau_{\ell})}\} \cup Z) \notag \\
    &\leq \sum_{\sigma \in \mathcal{S}_\ell}\sum_{\tau \in\mathcal{F}(\sigma) } \vol(\{u_{1}^{(\tau_1)}, \, u_{2}^{(\tau_2)}, \ldots, \,u_{\ell}^{(\tau_{\ell})}\} \cup Z) \,.
\end{align}

We will upper bound $\vol(\{u_{1}^{(\tau_1)}, \, u_{2}^{(\tau_2)}, \ldots, \,u_{\ell}^{(\tau_{\ell})}\} \cup Z)$ for each $\tau$ separately. As an illustration for a fixed $\sigma$, let us consider $\tau = (0,\ldots, 0, \sigma(k+1), \ldots, \sigma(\ell))$. The corresponding volume term is equal to
\begin{equation*}
    \frac{\vol(\{u_{1}^{(0)},\dots,u_{k}^{(0)},u_{k+1}^{(\tau_{k+1})},\dots,u_{\ell}^{(\tau_{\ell})} \} \cup Z)}{\vol(S)} = \frac{\vol(u_{1}^{(0)},\dots,u_{k}^{(0)})\cdot \vol(\{u_{k+1}^{(\tau_{k+1})},\dots,u_{\ell}^{(\tau_{\ell})}\} \cup Z)}{\vol(S)} \,,
\end{equation*}
as the sets of vectors $\{u_{1}^{(0)},\dots,u_{k}^{(0)}\}$ and $\{u_{k+1}^{(\tau_{k+1})},\dots,u_{\ell}^{(\tau_{\ell})}\} \cup Z$ are orthogonal to each other.  Upper bounding $\vol(u_{1}^{(0)},\dots,u_{k}^{(0)})$ by $\prod_{i=1}^{k}\norm{u_{i}^{\perp}}$, we get 
\begin{align*}
	 \frac{\vol(\{u_{1}^{(0)},\dots,u_{k}^{(0)},u_{k+1}^{(\tau_{k+1})},\dots,u_{\ell}^{(\tau_{\ell})} \} \cup Z)}{\vol(S)} \leq \left(\prod_{i=1}^{k}\norm{u_{i}^{\perp}}\right) \cdot\frac{\vol(\{u_{k+1}^{(\tau_{k+1})},\dots,u_{\ell}^{(\tau_{\ell})}\} \cup Z)}{\vol(S)}\, .
	\end{align*}  
To bound  $\vol(\{u_{k+1}^{(\tau_{k+1})},\dots,u_{\ell}^{(\tau_{\ell})}\} \cup Z)$, consider
\begin{align*}
	\frac{\vol(\{u_{k+1}^{(\tau_{k+1})},\dots,u_{\ell}^{(\tau_{\ell})}\} \cup Z)}{\vol(S)}&= 	\frac{\vol(\{a_{k+1,\tau_{k+1}} v_{\tau_{k+1}}, \ldots,a_{\ell,\tau_{\ell}} v_{\tau_{\ell}} \} \cup Z) }{\vol(S)} \\
	&= \left( \prod_{i=k+1}^{\ell}|a_{i\tau_i}| \right)\frac{\vol(\{v_{\sigma(k+1)},\dots, v_{\sigma(\ell)}\} \cup Z)}{\vol(S)} \,. \tag{since $\tau_j = \sigma(j)$ for $k <j \leq \ell$}
	\end{align*}
	Using~\ref{obs:log_vol_submodular}, 
	\begin{equation*}
	    \frac{\vol(\{v_{\sigma(k+1)},\dots, v_{\sigma(\ell)} \} \cup Z)}{\vol(S)} \leq \prod_{j\in [\ell]\backslash \{\sigma(k+1),\ldots, \sigma(\ell)\}} \frac{1}{\norm{v_{j}^{\perp}}} = \prod_{j\in  \{\sigma(1),\ldots, \sigma(k)\}} \frac{1}{\norm{v_{j}^{\perp}}}\,.
	\end{equation*}
Continuing the above chain of inequalities,
 \begin{align*}
 	\frac{\vol(\{u_{k+1}^{(\tau_{k+1})},\dots,u_{\ell}^{(\tau_{\ell})} \} \cup Z)}{\vol(S)} &\leq  \left( \prod_{i=k+1}^{\ell}|a_{i\tau_i}| \right) \cdot  \prod_{j\in \{\sigma(1),\ldots, \sigma(k)\}} \frac{1}{\norm{v_{j}^{\perp}}}  \cdot \prod_{i=1}^{k}\norm{u_{i}^{\perp}}\\
 	&=  \left( \prod_{i=k+1}^{\ell}|a_{i\tau_i}| \right) \cdot \left(\prod_{i=1}^{k} \frac{\norm{u_{i}^{\perp}}}{\norm{v_{\sigma(i)}^{\perp}}}\right) \;.
 \end{align*}
 
Now consider any $\tau \in \mathcal{F}(\sigma)$, and let $I_{\tau}^0 := \{j: \tau(j) = 0\}$ and $I_{\tau}^\sigma = [\ell]\backslash I_{\tau}^0$. Then following a similar chain of proof as above, we get
 \begin{equation*}
     \frac{\vol(\{u_{1}^{(\tau_1)},\dots,\dots,u_{\ell}^{(\tau_{\ell})}\} \cup Z)}{\vol(S)}  \leq \left(\prod_{i \in I_\tau^\sigma}|a_{i\sigma(i)}| \right) \cdot \left(\prod_{i \in I_\tau^0}\frac{\norm{u_{i}^{\perp}}}{\norm{v_{\sigma(i)}^{\perp}}} \right)\,.
 \end{equation*}
 Summing over all $\tau \in \mathcal{F}(\sigma)$, 
\begin{align*}
		\sum_{\tau \in \mathcal{F}(\sigma)}\frac{\vol(\{u_{1}^{(\tau_1)},\dots,\dots,u_{r}^{(\tau_{r})}\}\cup Z)}{\vol(S)}  &\leq \sum_{\tau \in \mathcal{F}(\sigma)} \left(\prod_{i \in I_\tau^\sigma}|a_{i\sigma(i)}| \right)\left(\prod_{i \in I_\tau^0}\frac{\norm{u_{i}^{\perp}}}{\norm{v_{\sigma(i)}^{\perp}}} \right) \\
		&= \sum_{\tau \in \mathcal{F}(\sigma)} \left(\prod_{i \in I_\tau^\sigma}e^{-\tilde{w}_0(u_i^\parallel, v_{\sigma(i)})} \right)\left(\prod_{i \in I_\tau^0}e^{-\tilde{w}_0(u_i^\perp, v_{\sigma(i)})} \right) \,.
	\end{align*}
Summing over all permutations, 
\begin{align*}
    r^{2r} \cdot \Tilde{f}(r) \leq \frac{\vol(T)}{\vol(S)}&\leq \sum_{\sigma \in \mathcal{S}_{\ell}}\sum_{\tau \in \mathcal{F}(\sigma)} \left(\prod_{i \in I_\tau^\sigma}e^{-\tilde{w}_0(u_i^\parallel, v_{\sigma(i)})} \right)\left(\prod_{i \in I_\tau^0}e^{-\tilde{w}_0(u_i^\perp, v_{\sigma(i)})} \right) .
\end{align*}
The RHS is sum of $2^\ell \ell!$ positive terms. So, there exists some permutation $\sigma \in \mathcal{S}_\ell$ and $\tau \in \mathcal{F}(\sigma)$ such that 
\begin{equation*}
    \left(\prod_{i \in I_\tau^\sigma}e^{-\tilde{w}_0(u_i^\parallel, v_{\sigma(i)})} \right)\left(\prod_{i \in I_\tau^0}e^{-\tilde{w}_0(u_i^\perp, v_{\sigma(i)})} \right) \geq \frac{r^{2r} \cdot \Tilde{f}(r)}{2^\ell \ell!} \geq \Tilde{f}(r)\,.
\end{equation*}
Let the cycle decomposition of $\sigma = \{\pi_1, \pi_2, \ldots, \pi_k\}$ with $\pi_i = \{j_1, j_2, \ldots, j_{x_i}\}$.
For each vector $u_i$, we define symbols $p_i$ indicating whether $u_i$ is present as a perpendicular vector in $\tau$, i.e., $p_i = \parallel$ if $i \in I_\tau^\sigma$ and $p_i = \perp$ otherwise.

Then each cyclic permutation $\pi_i$ corresponds to a cycle $C_i$ in $\widetilde{G}(S)$ given by
\begin{equation*}
    C_i = (u_{j_1}^{p_{j_1}} \rightarrow v_{j_2}  \rightarrow u_{j_2}^{p_{j_2}} \rightarrow v_{j_3} \ldots u_{j_{x_i}}^{p_{j_{x_{i}}}}\rightarrow v_{j_{1}} \rightarrow  u_{j_{1}}^{p_{j_1}}),
\end{equation*}
and for every $i \in [k]$, $ \left(\prod_{j \in I_\tau^\sigma \cap \pi_i}e^{-\tilde{w}_0(u_j^\parallel, v_{\sigma(j)})} \right) \cdot \left(\prod_{j \in I_\tau^0\cap \pi_i}e^{-\tilde{w}_0(u_j^\perp, v_{\sigma(j)})}\right) = e^{-\tilde{w}_0(C_i)}$ and therefore,
\begin{equation}
    \left(\prod_{i \in I_\tau^\sigma}e^{-\tilde{w}_0(u_i^\parallel, v_{\sigma(i)})} \right)\cdot \left(\prod_{i \in I_\tau^0}e^{-\tilde{w}_0(u_i^\perp, v_{\sigma(i)})} \right) = \prod_{i=1}^k e^{-\tilde{w}_0(C_i)} \geq \Tilde{f}(r). \label{eq:t1}
\end{equation}
At least one of the $C_i$'s must be an $\tilde{f}$-violating cycle. If not, then $e^{-\tilde{w}_0(C_i)} \leq \Tilde{f}(|C_i|/2)$ and \begin{equation}
     \left(\prod_{i \in I_\tau^\sigma}e^{-\tilde{w}_0(u_i^\parallel, v_{\sigma(i)})} \right)\left(\prod_{i \in I_\tau^0}e^{-\tilde{w}_0(u_i^\perp, v_{\sigma(i)})} \right) = \prod_{i=1}^k e^{-\tilde{w}_0(C_i)} \leq  \prod_{i=1}^k \Tilde{f}(|C_i|/2) \leq \Tilde{f}(r), \label{eq:t2}
\end{equation} 
where the last inequality follows from $\sum_{i=1}^k |C_i| = 2\ell \leq 2r$. Equation~\eqref{eq:t2} contradicts~\eqref{eq:t1}, so there exist $i$ such that $C_i$ is an $\tilde{f}$-violating cycle in $\widetilde{G}(S)$.
\end{proof}

\subsection{Analysis of Stage 2}

In this section, we prove that if the algorithm finds an $\tilde{f}$-violating cycle in Stage 2, then exchanging on this cycle increases the volume by a constant factor. However, since the algorithm fails to find a cycle in Stage 1, a cycle in Stage 2 must contain a perpendicular vertex. We first bound the number of vertices a minimal $\tilde{f}$-violating cycle can contain, and use this fact to prove that exchanging on such a cycle increases the volume.

\begin{lemma}
	If $C$ is a minimal $\Tilde{f}$-violating cycle in $\widetilde{G}(S)$ and $\widetilde{C} = \{ u \in [n]\del S: u^\perp \text{ or } u^\parallel \in C\} \cup \{ v \in S: v \in C\}$, then $S\Delta \widetilde{C}$ is independent in $\mathcal{M}$.
\end{lemma}
\begin{proof}
	Observation~\ref{obs:parallel-matroid} and Lemma~\ref{lem:matroid-indep} imply that $S \Delta \widetilde{C}$ is independent in $\widetilde{\mathcal{M}}(S)$, and the second part of Observatio ~\ref{obs:parallel-matroid} then implies that $S\Delta \widetilde{C}$ is independent in $\mathcal{M}$.
\end{proof}

To make our later calculations possible, we first prove that any minimal $\Tilde{f}$-violating cycle in $\widetilde{G}(S)$ contains exactly one perpendicular vector. The following lemma is slightly more general than we need now, but will be useful later.

\begin{lemma} \label{lem:parts}
Let $C$, $|C| = 2y$ be a cycle in $\widetilde{G}(S)$ such that there are no $\Tilde{f}$-violating cycles with less than $2y$ hops in $\widetilde{G}(S)$. Let $C$ contain $k \geq 2$ perpendicular vertices $u_1^\perp, u_2^\perp, \ldots, u_k^\perp$ such that the section of $C$ from $u_i^\perp$ to  $u_{(i \; \mathrm{mod} \;  k)+1}^\perp$ has $2x_i$ hops for $i \in [k]$. Then \begin{equation*}
    e^{-\Tilde{w}_0(C)} \leq \prod_{i=1}^k \Tilde{f}(x_i).
\end{equation*}
\end{lemma}

\begin{proof}
	The proof is by induction on $k>1$. Let the perpendicular vertices $u_1^\perp, u_2^\perp, \ldots, u_k^\perp$ appear in order along around $C$. Let $C'$ and $C''$ be the two cycles created by replacing the edges $u_1^\perp\rightarrow v_1$ and $u_2^\perp \rightarrow v_2$ with the new edges $u_1^\perp \rightarrow v_2$ and $u_2^\perp \rightarrow v_1$, so that $C'$ is the cycle containing $u_1^\perp$. Note that $C'$ has $2y - 2x_1 < 2y$ hops and $C''$ has $2x_1 < 2y$ hops. So by the hypothesis of the Lemma, both $C'$ and $C''$ are not $\tilde{f}$-violating. 
	
    Additionally,
	\begin{equation*}
	    \Tilde{w}_0(u_1^\perp \rightarrow v_2) + \Tilde{w}_0(u_2^\perp \rightarrow v_1) = -\log\left( \frac{\|u_1^\perp\| \|u_2^\perp\|}{\|v_1^\perp\|\|v_2^\perp\|}\right) = \Tilde{w}_0(u_1^\perp\rightarrow v_1) + \Tilde{w}_0(u_2^\perp \rightarrow v_2),
	\end{equation*}
	and therefore
	\begin{equation}
	    e^{-\Tilde{w}_0(C)} = e^{-\Tilde{w}_0(C') -\Tilde{w}_0(C'')}. \label{eq:e7}
	\end{equation}

	 When $k=2$, both cycles $C'$ and $C''$ contain exactly one perpendicular vertex, and since they are not $\tilde{f}$-violating, $e^{-\tilde{w}_0(C')} \leq \tilde{f}(2y-2x_1)$ and $e^{-\tilde{w}_0(C'')} \leq \tilde{f}(2x_1)$. 
	 
	Using equation~\eqref{eq:e7}, we conclude that $e^{-\Tilde{w}_0(C)} = e^{-\Tilde{w}_0(C') -\Tilde{w}_0(C'')} \leq \tilde{f}(2y-2x_1) \cdot \Tilde{f}(x_1)$, as desired.
	
	When $k > 2$ the cycle $C'$ has $k-1 \geq 2$ perpendicular vertices since it no longer contains $u_2^\perp$. Thus, the induction hypothesis implies that $e^{-\Tilde{w}_0(C')} \leq \prod_{i=2}^k \Tilde{f}(x_i)$. Since $C''$ is not $\tilde{f}$-violating, it satisfies $e^{-\tilde{w}_0(C'')} \leq \tilde{f}(2x_1)$. 
	
	Again using equation~\eqref{eq:e7}, we conclude that $e^{-\Tilde{w}_0(C)} = e^{-\Tilde{w}_0(C') -\Tilde{w}_0(C'')} \leq  \prod_{i=1}^k \Tilde{f}(x_i)$, as desired.
\end{proof}

When $y = \sum_{i=1}^k x_i$, we know that $\Tilde{f}(y) \geq \prod_{i=1}^k \Tilde{f}(x_i)$. Thus, we obtain the following corollary:

\begin{corollary} \label{cor:one}
	If $C$ is a minimal $\Tilde{f}$-violating cycle in $\widetilde{G}(S)$, then $C$ contains exactly one perpendicular vector.
\end{corollary}
Now we will prove that for a minimal $\tilde{f}$-violating cycle $C$, the volume of $S\Delta C$ is strictly larger than $\vol(S)$.

\begin{lemma}
If $C$ is a minimal $\Tilde{f}$-violating cycle in $\widetilde{G}(S)$, then $\vol(S \Delta C) \geq 2 \cdot \vol(S)$.
\end{lemma}
\begin{proof}
If $|C| = 2$, then let $C = (u^\perp \rightarrow v)$. Since $C$ is an $\Tilde{f}$-violating cycle, $e^{-\Tilde{w}_0(C)} = \frac{\norm{u^\perp}}{\norm{v^\perp}} > \Tilde{f}(1) \geq 2$. By Lemma \ref{lem:weight_wa0}, \begin{equation*}
\frac{\vol(S-v+u)}{\vol(S)} =  \sqrt{a_{uv}^2 + \frac{\norm{u^\perp}^2}{\norm{v^\perp}^2} }\geq e^{-\Tilde{w}_0(C)} \geq 2.
\end{equation*}
This concludes the proof when $C$ has $2$ arcs.

Now consider the case when $C$ has at least $4$ arcs. By Corollary~\ref{cor:one}, we can assume that $C$ contains exactly one perpendicular vertex.
Let $C = ( u_1^\perp \rightarrow v_1 \rightarrow u_2^\parallel \rightarrow v_2\rightarrow \ldots  u_{\ell}^\parallel \rightarrow v_{\ell} \rightarrow u_1^\perp)$, where $\ell \geq 2$, $v_i \in S$, and $v_i \rightarrow u_{i+1}$ is an arc in $G(S)$.

Define $T := S \Delta C$, $X := C \backslash S = \{u_1, u_2, \ldots, u_\ell\}$, and $Y := S \cap C = \{v_1, v_2, \ldots, v_\ell \}$, and let $S\backslash Y = \{v_{\ell+1}, \ldots, v_{r}\}$. Let $\proj$ denote the projection matrix orthogonal to the span of $S\backslash Y$.

Then \begin{align*}
    \frac{\vol(T)}{\vol(S)} =  \frac{\vol(\{S \backslash Y\} \cup X)}{\vol(\{S\backslash Y\} \cup Y)} = \frac{\vol(u_1, u_2, \ldots,  u_\ell, v_{\ell+1}, \ldots, v_{r})}{\vol( v_1,  v_2, \ldots, v_\ell, v_{\ell+1}, \ldots,  v_r)} \,.
\end{align*}
Note that for any set of vectors $\{x_1, \ldots, x_\ell\}$, 
\[\vol(x_1, x_2, \ldots, x_\ell, v_{\ell+1}, \ldots, v_{r}) = \vol(\proj x_1, \proj x_2, \ldots, \proj x_\ell) \cdot \vol(v_{\ell+1}, \ldots, v_{r}).\]
Applying this to the sets $X$ and $Y$, we get
\begin{align*}
    \frac{\vol(T)}{\vol(S)} =  \frac{\vol(\proj u_1, \proj u_2, \ldots, \proj u_\ell)}{\vol(\proj v_1, \proj v_2, \ldots, \proj v_\ell)} =  \frac{\vol(\proj X)}{\vol(\proj Y)}\;.
\end{align*}

By definition, $u_i = \sum_{j=1}^r a_{i,j} v_j + u_i^\perp$. Taking the projection of $u_i$ orthogonal to $\Span(S\backslash Y)$, we get $\proj u_i = \sum_{j = 1}^{\ell} a_{i,j} \proj v_j + u_i^\perp$, since $Pv_{j} = 0$ for all $j > \ell$.
So we can decompose each vector $\proj u_i$ into a sum of $\ell+1$ vectors, i.e., $\proj u_i = \sum_{j=0}^\ell u_{i}^{(j)}$
where
\begin{itemize}
    \item $u_{i}^{(0)} := u_i^\perp$ for all $i$, and
    \item $u_{i}^{(j)} := a_{i, j} \proj v_j$.
\end{itemize}
Note that for any $i_1, i_2 \in [\ell]$, the vectors $u_{i_1}^{(j)}$ and $u_{i_2}^{(j)}$ are linearly dependent for any $j > 0$.

Let $I_C$ denote the tuple $(0,2,3, \ldots, \ell)$. Using Observation \ref{obs:volume_triangle}, we can lower bound $\vol(\proj X)$ as
\begin{equation}
    \vol(\proj X) \geq \vol(\,u_{1}^{(0)}, \, u_{2}^{(2)}, \ldots, \,u_{\ell}^{(\ell)}) \; - \sum_{\substack{i_1, i_2, \ldots, i_\ell \in \{0,1,\dots,\ell\} \\ (i_1, i_2, \ldots, i_\ell) \neq I_C}} \vol(\,u_{1}^{(i_1)}, \, u_{2}^{(i_2)}, \ldots, \,u_{\ell}^{(i_\ell)})\,. \label{eq:vl}
\end{equation}
For any permutation $\sigma \in \mathcal{S}_\ell$, define $\mathcal{F}(\sigma) := \{(i_1, i_2, \ldots, i_\ell): i_j \in \{\sigma(j), 0\} \}$.
Since $u_{i_1}^{(j)}$ and $u_{i_2}^{(j)}$ are linearly dependent, we can rewrite~\eqref{eq:vl} as
\begin{align}
    \vol(\proj X) &\geq \vol(\,u_{1}^{(0)}, \, u_{2}^{(2)}, \ldots, \,u_{\ell}^{(\ell)}) \; - \sum_{\tau \in \bigcup_{\sigma} \mathcal{F}(\sigma) \backslash I_C} \vol(\,u_{1}^{(\tau_1)}, \, u_{2}^{(\tau_2)}, \ldots, \,u_{\ell}^{(\tau_\ell)}) \notag \\
    &\geq \vol(\,u_{1}^{(0)}, \, u_{2}^{(2)}, \ldots, \,u_{\ell}^{(\ell)}) \; - \sum_{\sigma \in S_\ell}\sum_{\tau \in\mathcal{F}(\sigma) \backslash I_C} \vol(\,u_{1}^{(\tau_1)}, \, u_{2}^{(\tau_2)}, \ldots, \,u_{\ell}^{(\tau_\ell)})\,. \label{eq:a7}
\end{align}
For a fixed $\sigma \neq id_\ell$ and some $\tau \in \mathcal{F}(\sigma)$, let $I_{\tau}^0 := \{j \in [\ell]: \tau(j) = 0\}$ and $I_{\tau}^\sigma := [\ell]\backslash I_{\tau}^0$. Also, let $X_\tau = \{v_{\sigma(i)}: i \in I_\tau^\sigma\}$. Then
\begin{align*}
   \frac{\vol(\,u_{1}^{(\tau_1)}, \, u_{2}^{(\tau_2)}, \ldots, \,u_{\ell}^{(\tau_\ell)})}{\vol(\proj Y)} &= \frac{\vol\left( \bigcup_{i \in I_\tau^0} \, \{ u_{i}^{\perp} \}\right) \cdot \vol\left( \bigcup_{i \in I_\tau^\sigma} \, \{ a_{i, \sigma(i)} \proj v_{\sigma(i)})\}\right)}{\vol(\proj Y)}\\
   &\leq \prod_{i \in I_\tau^0} \norm{u_i^\perp} \cdot \prod_{i \in I_\tau^\sigma} |a_{i,\sigma(i)}| \cdot \frac{\vol\left( \bigcup_{i \in I_\tau^\sigma} \, \{  \proj v_{\sigma(i)})\}\right)}{\vol(\proj Y)} \\
   &\leq \prod_{i \in I_\tau^0} \norm{u_i^\perp} \cdot \prod_{i \in I_\tau^\sigma} |a_{i,\sigma(i)}| \cdot \frac{1}{\prod_{i \in I_\tau^0} \norm{v_{\sigma(i)}^\perp}} \tag{from Observation \ref{obs:log_vol_submodular}}\\
   &=  \prod_{i \in I_\tau^0} \frac{\norm{u_i}}{\norm{v^\perp_{\sigma(i)}}} \cdot \prod_{i \in I_\tau^\sigma} |a_{i,\sigma(i)}| = \prod_{i \in I_\tau^0} e^{-\Tilde{w}_0(u_i^\perp, v_{\sigma(i)})} \cdot \prod_{i \in I_\tau^\sigma} e^{-\Tilde{w}_0(u_i^\parallel, v_{\sigma(i)})}\,.
\end{align*}
Summing over all tuples $\tau \in \mathcal{F}(\sigma)$, we get
\begin{align}
    \sum_{\tau \in \mathcal{F}(\sigma)} \frac{\vol(\,u_{1}^{(\tau_1)}, \, u_{2}^{(\tau_2)}, \ldots, \,u_{\ell}^{(\tau_\ell)}) }{\vol(\proj Y)}  &\leq \prod_{i \in I_\tau^0} e^{-\Tilde{w}_0(u_i^\perp, v_{\sigma(i)})} \cdot \prod_{i \in I_\tau^\sigma} e^{-\Tilde{w}_0(u_i^\parallel, v_{\sigma(i)})} \notag \\
    &= \prod_{i = 1}^\ell \left(e^{-\Tilde{w}_0(u_i^\perp, v_{\sigma(i)}) }+ e^{-\Tilde{w}_0(u_i^\parallel, v_{\sigma(i)})}\right)\,. \label{eq:a8}
\end{align}
Let $W_C \in \mathbb{R}^{\ell \times \ell}$ be a matrix with $[W_C]_{i,j} = e^{-\Tilde{w}_0(u_i^\perp, v_{\sigma(i)}) }+ e^{-\Tilde{w}_0(u_i^\parallel, v_{\sigma(i)})}$. Then we can rewrite~\eqref{eq:a8} as
\begin{align}
    \sum_{\tau \in \mathcal{F}(\sigma)} \frac{\vol(\,u_{1}^{(\tau_1)}, \, u_{2}^{(\tau_2)}, \ldots, \,u_{\ell}^{(\tau_\ell)}) }{\vol(\proj Y)}  &\leq \prod_{i = 1}^\ell [W_C]_{i, \sigma(i)}\,. \label{eq:a13}
\end{align}
Similarly, the sum of tuples for the identity permutation $id_\ell$ gives
\begin{align}
   \frac{\vol(\,u_{1}^{(0)}, \, u_{2}^{(2)}, \ldots, \,u_{\ell}^{(\ell)})}{\vol(\proj Y)} &+ \sum_{\tau \in \mathcal{F}(id_\ell) \backslash I_C} \frac{\vol(\,u_{1}^{(\tau_1)}, \, u_{2}^{(\tau_2)}, \ldots, \,u_{\ell}^{(\tau_\ell)}) }{\vol(\proj X)}  \leq \prod_{i = 1}^\ell [W_C]_{i, i}\,. \label{eq:a9}
\end{align}
Plugging equation~\eqref{eq:a13} and equation~\eqref{eq:a9} in~\eqref{eq:a7}, we get
\begin{align}
    \frac{\vol(\proj X)}{\vol(\proj Y)} &\geq \frac{\vol(\,u_{1}^{(0)}, \, u_{2}^{(2)}, \ldots, \,u_{\ell}^{(\ell)})}{\vol(\proj Y)} -  \sum_{\sigma \neq {id}}\prod_{i = 1}^\ell [W_C]_{i, \sigma(i)} - \left(\prod_{i = 1}^\ell [W_C]_{i, i} - \frac{\vol(\,u_{1}^{(0)}, \, u_{2}^{(2)}, \ldots, \,u_{\ell}^{(\ell)})}{\vol(\proj Y)} \right) \notag \\
    &= \frac{2\vol(\,u_{1}^{(0)}, \, u_{2}^{(2)}, \ldots, \,u_{\ell}^{(\ell)})}{\vol(\proj Y)} -  \sum_{\sigma \neq {id}}\prod_{i = 1}^\ell [W_C]_{i, \sigma(i)} - \prod_{i = 1}^\ell [W_C]_{i, i} \notag  \\
    &= \frac{2\vol(\,u_{1}^{(0)}, \, u_{2}^{(2)}, \ldots, \,u_{\ell}^{(\ell)})}{\vol(\proj Y)} -  \perm(W_C) \,. \label{eq:a10}
\end{align}
From Lemma \ref{lem:permw}, we have
\begin{equation}
    \perm(W_C) \leq 1.56 \cdot \frac{\norm{u_1^\perp}}{\norm{v_1^\perp}} \cdot \prod_{i=2}^\ell |a_{i, i}|\,. \label{eq:a11}
\end{equation}
Note that \begin{equation}
     \frac{\vol(\,u_{1}^{(0)}, \, u_{2}^{(2)}, \ldots, \,u_{\ell}^{(\ell)})}{\vol(\proj Y)} = \frac{\norm{u_1^\perp}}{\norm{v_1^\perp}} \cdot \prod_{i=2}^\ell |a_{i, i}|\,. \label{eq:a12}
\end{equation}
Inserting the bounds from~\eqref{eq:a11} and~\eqref{eq:a12} in~\eqref{eq:a10} gives
\begin{align*}
    \frac{\vol(\proj X)}{\vol(\proj Y)}
    &\geq \frac{\norm{u_1^\perp}}{\norm{v_1^\perp}} \cdot \prod_{i=2}^\ell |a_{i, i}|\cdot (2 - 1.56) = 0.44 \cdot \frac{\norm{u_1^\perp}}{\norm{v_1^\perp}} \cdot \prod_{i=2}^\ell  |a_{i, i}|\,.
\end{align*}
Since $C$ is an $\Tilde{f}$-Violating cycle, $e^{-\Tilde{w}_0(C)} = \frac{\norm{u_1^\perp}}{\norm{v_1^\perp}} \cdot \prod_{i=2}^\ell |a_{i, i}| \geq \Tilde{f}(\ell)$. Therefore, \begin{align*}
   \frac{\vol(S\Delta C)}{ \vol(S)} = \frac{\vol(\proj X)}{\vol(\proj Y)}
    &\geq 0.44 \cdot \Tilde{f}(\ell) \geq 2 \,.
\end{align*}
\end{proof}

\subsection{Miscellaneous Lemmas}
\begin{lemma} \label{lem:permw}
Let $C =  (v_0 \rightarrow u^\perp_1 \rightarrow v_1 \rightarrow u^\parallel_2 \rightarrow v_2\rightarrow \ldots  u^\parallel_{\ell} \rightarrow v_{0})$ be a minimal $\Tilde{f}$-violating cycle in $\widetilde{G}(S)$ with $\ell \geq 2$ and let $W_C$ be a matrix with $[W_C]_{i,j} = |a_{i,j}| + \frac{\norm{u_i^\perp}}{\norm{v_j^\perp}}$, then \begin{equation*}
   \perm(W_C) \leq 1.56 \cdot \frac{\norm{u_1^\perp}}{\norm{v_1^\perp}} \cdot \prod_{i=2}^\ell |a_{i, i}| \,.
\end{equation*}
\end{lemma}
\begin{proof}
Define $z_1 := \frac{\norm{u_1^\perp}}{\norm{v_1^\perp}}$, $p_1 = \perp$, and $z_i := |a_{i,i}|$, $p_i = \parallel$ for all $i \geq 2$.

We show upper bounds on the absolute value of each entry of $W_C$ as a function of $z_i$'s. Consider the $i,j$-th entry of $W_C$. For $i > j$, let $C_{i,j}^\perp := (u_i^\perp \rightarrow v_j \rightarrow u_{j+1}^{p_{j+1}} \rightarrow v_{j+1}\rightarrow \ldots  v_{i-1} \rightarrow u_i^{\perp})$. $C_{i,j}^{\perp}$ is a cycle of with $2(i-j)$ hops and $V(C_{i,j}^\perp) \subset V(C)$. $C$ being a minimal $\Tilde{f}$-violating cycle implies that $C_{i,j}^\perp$ is not an $\Tilde{f}$-violating cycle. Therefore, $e^{-\tilde{w}_0(C_{i,j}^\perp)} = \frac{\norm{u_i^\perp}}{\norm{v_j^\perp}} \cdot \prod_{s = j+1}^{i-1} z_s < \Tilde{f}(i-j)$.
This implies
\begin{equation}
    \frac{\norm{u_i^\perp}}{\norm{v_j^\perp}} < \frac{\Tilde{f}(i-j)}{\prod_{s = j+1}^{i-1} z_s}\,.
\end{equation}
Similarly for $j = \ell$, we have $ \frac{\norm{u_i^\perp}}{\norm{v_0^\perp}} < \frac{\Tilde{f}(i)}{\prod_{s = 1}^{i-1} z_s}\,.$

Similarly, let $C_{i,j}^\parallel := (u_i^\parallel \rightarrow v_j \rightarrow u_{j+1}^{p_{j+1}} \rightarrow v_{j+1}\rightarrow \ldots  v_{i-1} \rightarrow u_i^{\parallel})$. Again, using the fact that $C_{i,j}^\parallel$ is not $\Tilde{f}$-Violating we get \begin{equation}
    |a_{i, j}| < \frac{\Tilde{f}(i-j)}{\prod_{s = j+1}^{i-1} z_s} \,,
\end{equation}
for any $i > j$ and for $j = \ell$, we have $ a_{i, 0} < \frac{\Tilde{f}(i)}{\prod_{s = 1}^{i-1} z_s}\,.$

For $i < j < \ell$, let $C^\perp_{i,j} := (v_0 \rightarrow u_1^{p_i} \rightarrow v_1\rightarrow \ldots u_i^{\perp} \rightarrow v_j\rightarrow \ldots  u_{\ell}^{p_{\ell}} \rightarrow v_{0})$. Again, $C^\perp_{i,j}$ is a cycle with $2(\ell - j+i)$ hops and which is not $\Tilde{f}$-violating. Therefore, \begin{equation}
   \frac{\norm{u_i^{\perp}}}{\norm{v_j^{\perp}}} \cdot \prod_{s = 1}^{i-1} z_s \cdot \prod_{s = j+1}^{\ell} z_s < \Tilde{f}(\ell - j + i). \label{eq:2}
\end{equation}
Since $C$ is an $\Tilde{f}$-violating cycle, we also have \begin{equation}
    e^{-\tilde{w}_0(C)} = \prod_{s = 1}^\ell z_s > \Tilde{f}(\ell)\,. \label{eq:4}
\end{equation}
Combining~\eqref{eq:2} and~\eqref{eq:4} gives
 \begin{align*}
     \frac{\norm{u_i^{\perp}}}{\norm{v_j^{\perp}}} &< \frac{\Tilde{f}(\ell-j+i)}{\Tilde{f}(\ell)} \cdot \prod_{s=i}^j z_s \,.
\end{align*}
Similarly, for $i < j < \ell$,
 \begin{align*}
     |a_{i, j}| &< \frac{\Tilde{f}(\ell-j+i)}{\Tilde{f}(\ell)} \cdot \prod_{s=i}^j z_s \,.
\end{align*}\
Define $x_1 := (|a_{1,1}| + \frac{\norm{u_1^\perp}}{\norm{v_1^\perp}})/ \frac{\norm{u_1^\perp}}{\norm{v_1^\perp}} $ and $x_i :=(|a_{i,i}| + \frac{\norm{u_i^\perp}}{\norm{v_i^\perp}})/ |a_{i,i}|$ for $i \geq 2$. Then $x_i \geq 1$ and the $i$-th diagonal entry of $W_C$ is $x_i \cdot z_i$.

Let $B_{\ell}$ be the matrix obtained by applying the following operations to $W_C$ 
\begin{itemize}
    \item Multiply the last column by $z_1$ and for $j > 1$, divide the $j$-th column by $\prod_{s=2}^{j} z_s$
    \item Divide the first row by $z_1$ and for $i > 1$, multiply the $i$-th row by $\prod_{s = 2}^{i-1} z_s$
    \item Divide the last column by $\prod_{i=1}^\ell z_i$.
    \item Divide the $j$-th column by $x_j$.
\end{itemize}
Then $|\perm(W_C)| = (\prod_{i=1}^\ell x_i z_i) \cdot |\perm(B_{\ell})|$, and $B_\ell$ satisfies the following properties:
\begin{itemize}
    \item $b_{i,i}$ = 1 for all $i \in [\ell]$,
    \item $|b_{i,j}| \leq 2 \cdot \Tilde{f}(i-j)/x_j \leq 2\cdot  \Tilde{f}(i-j)$ for all $j < i \leq \ell$, and
    \item $|b_{i,j}| \leq 2\cdot \Tilde{f}(\ell - j + i)/(\Tilde{f}(\ell)x_j) \leq 2\cdot \Tilde{f}(\ell - j + i)/\Tilde{f}(\ell)$ for all $i < j \leq \ell$.
\end{itemize}
If $\ell = 1$, $\perm(B_\ell) = 1$ and for $\ell \geq 2$, Lemma \ref{lem:perm} gives $\perm(B_\ell) \leq 1.3$. Therefore, we have \begin{equation*}
    \perm(W_C) \leq 1.3 \cdot \prod_{i=1}^\ell x_i z_i = 1.3\cdot  \prod_{i=1}^\ell \left(|a_{i,i}| + \frac{\norm{u_i^\perp}}{\norm{v_i^\perp}}\right) \leq 1.3 \cdot 1.2 \cdot \frac{\norm{u_1^\perp}}{\norm{v_1^\perp}} \cdot \prod_{i=2}^\ell |a_{i, i}| = 1.56 \cdot e^{-\tilde{w}_0(C)}\;,
\end{equation*}
where the last inequality follows from Lemma \ref{lem:diag}.
\end{proof}

\begin{lemma}\label{lem:diag}
If $C =  (v_0 \rightarrow u^\perp_1 \rightarrow v_1 \rightarrow u^\parallel_2 \rightarrow v_2\rightarrow \ldots  u^\parallel_{\ell} \rightarrow v_{0})$ is a minimal $\Tilde{f}$-violating cycle with $\ell \geq 2$, then \begin{equation}
    \prod_{i=1}^\ell \left(|a_{i,i}| + \frac{\norm{u_i^\perp}}{\norm{v_i^\perp}}\right) \leq 1.2 \cdot e^{-\tilde{w}_0(C)}. \label{eq:r1}
\end{equation}
\end{lemma}
\begin{proof}
Since $C$ is an $\Tilde{f}$-violating cycle, \begin{equation}
    e^{-\Tilde{w}_0(C)} = \frac{\norm{u_1^\perp}}{\norm{v_1^\perp}} \cdot \prod_{i=2}^\ell |a_{i, i}| \geq \Tilde{f}(\ell). \label{eq:a4}
\end{equation}

Define $x_{i, 0} := \frac{\norm{u_i^\perp}}{\norm{v_i^\perp}}$ and $x_{i, 1} := |a_{i,i}|$ for all $i \in [\ell]$. Then we can decompose the LHS of~\eqref{eq:r1} as follows:
\begin{align*}
      \prod_{i=1}^\ell \left(|a_{i,i}| + \frac{\norm{u_i^\perp}}{\norm{v_i^\perp}}\right) = \sum_{Z \subseteq [\ell]} \prod_{i \in Z}x_{i, 0} \cdot \prod_{i \notin Z} x_{i, 1} =: \sum_{Z \subseteq [\ell]}  w(Z) ,
\end{align*}
We now upper bound $w(Z)$ for all $Z \neq \{1\}$ based on the cardinality of $Z$.  
For the empty set, i.e., $|Z| = 0$, since Algorithm \ref{alg:exchr} did not return a cycle in Stage 1, $w(Z) =  \prod_{i=1}^\ell |a_{i,i}| = e^{-w_0(C)} \leq f(\ell)$.

Consider a subset $Z$ with $|Z| = k$ where $k \geq 2$. We define a cycle $C_Z$ such that $C_Z$ contains the same arcs as $C$ but a vertex $u_i \in C\backslash S$ is present as a perpendicular vertex in $C_Z$ if and only if $u \in Z$. Then $w(Z) = e^{-\Tilde{w}_0(C_Z)}$. Let the distance between successive perpendicular vertices in $C_Z$ be $b_1, b_2, \ldots, b_k$. Then $\sum_{i=1}^k b_i = \ell$ and by Lemma \ref{lem:parts},  $w(Z) =  e^{-\Tilde{w}_0(C_Z)} \leq \prod_{i=1}^k \Tilde{f}(b_i)$.

Since $C_Z$ has the same structure as $C$, we can completely characterize $C_Z$ by the first vertex among $u_1, u_2, \ldots, u_\ell$ which is perpendicular, and the distances between successive perpendicular vertices in $C_Z$, i.e., $b_1, b_2, \ldots, b_k$. There are only $\ell$ options for the first perpendicular vertex. Therefore,
\begin{align}
   \sum_{Z:|Z| = k} w(Z) \leq \ell \cdot \sum_{\substack{\sum_{j=1}^k b_j = \ell \\
   b_1, b_2, \ldots, b_k \geq 1}} \prod_{i=1}^k \Tilde{f}(b_i) \leq \ell \cdot \binom{\ell-1}{k-1}\cdot \Tilde{f}(\ell-k+1) \leq \frac{\Tilde{f}(\ell)}{\ell^8}\,. \label{eq:a5}
\end{align}

For a subset with one element, i.e., $Z = \{x\}$, $x \neq 1$, $-\log w(Z)$ is the same as the weight of the cycle
\begin{equation*}
    C_x := (v_0 \rightarrow u_1^\parallel \rightarrow v_1 \rightarrow \ldots u_x^\perp \rightarrow v_2 \rightarrow \ldots \rightarrow u_\ell^\parallel \rightarrow v_0) \,.
\end{equation*}
Combining $C_x$ with $C$, we get two cycles
\begin{align*}
    C_x^{(1)} &= (v_0 \rightarrow u_1^\perp \rightarrow v_1 \rightarrow \ldots u_x^\perp \rightarrow v_2 \rightarrow \rightarrow \ldots u_\ell^\parallel \rightarrow v_0) \\
     C_x^{(2)} &= (v_0 \rightarrow u_1^\parallel \rightarrow v_1 \rightarrow \ldots u_x^\parallel \rightarrow v_2 \rightarrow \ldots \rightarrow u_\ell^\parallel \rightarrow v_0)\,, 
\end{align*}
such that $\Tilde{w}_0(C) \cdot \Tilde{w}_0(C_x) = \Tilde{w}_0(C_x^{(1)}) \cdot \Tilde{w}_0(C_x^{(2)}))$.
By Lemma \ref{lem:parts}, we know that $e^{-\Tilde{w}_0(C_x^{(1)})} \leq \Tilde{f}(x-1)\Tilde{f}(\ell-x+1)$. Also, since $C_x^{(2)}$ only contains parallel vectors, $\Tilde{w}_0(C_x^{(2)}) =w_0(C_x^{(2)})$; and since $C_x^{(2)}$ is not $f$-violating, $e^{-\Tilde{w}_0(C_x^{(2)})} \leq f(\ell)$.
Therefore, $e^{-\Tilde{w}_0(C) - \Tilde{w}_0(C_x)} \leq \Tilde{f}(x-1)\Tilde{f}(\ell-x+1) \cdot f(\ell)$. $C$ being $\tilde{f}$-violating implies $e^{-\Tilde{w}_0(C)} \geq \Tilde{f}(\ell)$, and therefore $e^{-\Tilde{w}_0(C)} \leq  \Tilde{f}(x-1)\Tilde{f}(\ell-x+1) f(\ell)/\Tilde{f}(\ell)$.

Summing over all choices of $x$, we get
\begin{equation}
   \sum_{Z, |Z| = 1} w(Z) =  \sum_{x=2}^\ell e^{-\Tilde{w}_0(C_x)} \leq \sum_{x = 2}^\ell \frac{\Tilde{f}(x-1)\Tilde{f}(\ell-x+1) f(\ell)}{\Tilde{f}(\ell)} \leq \frac{4f(\ell)}{\ell} \,. \label{eq:a6}
\end{equation}
Combining~\eqref{eq:a4}, ~\eqref{eq:a5}, and~\eqref{eq:a6}, we get
\begin{align*}
    \prod_{i=1}^\ell \left(|a_{i,i}| + \frac{\norm{u_i^\perp}}{\norm{v_i^\perp}}\right) &= \sum_{Z \subseteq [\ell]}  w(Z) \leq w(\emptyset) + \sum_{Z: |Z| = 1} w(Z) + \sum_{Z: |Z| \geq 2} w(Z) \\
    &\leq f(\ell) + e^{-\Tilde{w}_0(C)} + \frac{4 f(\ell)}{\ell}  + \frac{\Tilde{f}(\ell)}{\ell^8}\,.
\end{align*}
Since $e^{-\Tilde{w}_0(C)} \geq \Tilde{f}(\ell) \geq (\ell !)^4 \cdot f(\ell)$,
\begin{align*}
    \prod_{i=1}^\ell \left(|a_{i,i}| + \frac{\norm{u_i^\perp}}{\norm{v_i^\perp}}\right) &\leq e^{-\Tilde{w}_0(C)} \left(1+ \frac{3}{(\ell !)^4} + \frac{1}{\ell^8}\right) \leq 1.2 \cdot e^{-\Tilde{w}_0(C)}.
\end{align*}
\end{proof}

\begin{lemma} \label{lem:perm}
Let $B_\ell \in \mathbb{R}^{\ell \times \ell}$ satisfy the following properties:
\begin{itemize}
    \item $b_{i,i}$ = 1 for all $i \in [\ell]$,
    \item $0 \leq b_{i,j} \leq  2 \cdot \Tilde{f}(i-j)$ for all $j < i \leq \ell$, and
    \item $0 \leq b_{i,j} \leq 2 \cdot \Tilde{f}(\ell - j + i)/\Tilde{f}(\ell)$ for all $i < j \leq \ell$.
\end{itemize}
Then the permanent of $B_\ell$ is at most $1.13$.
\end{lemma}
\begin{proof}
For $\ell = 2$, $\perm(B_2) = 1 + b_{1,2} \cdot b_{2, 1} \leq 1 + \frac{4\tilde{f}(1)^2}{\tilde{f}(2)} < 1.1$. 

Now, consider the case when $\ell \geq 3$. Let $id_\ell$ denote the identity permutation on $\ell$ elements. Expanding the permanent of $B_\ell$ gives \begin{equation}
    \perm(B_\ell) =\sum_{\sigma \in \mathcal{S}_\ell} \prod_{i=1}^\ell b_{i, \sigma(i)} = 1 + \sum_{\sigma \in \mathcal{S}_\ell \backslash \{id_\ell\}} \prod_{i=1}^\ell b_{i, \sigma(i)} \,. \label{eq:d1}
\end{equation}
We categorize all permutations in $\mathcal{S}_\ell \backslash \{id_\ell\}$ based on the number of fixed points and the number of exceedances. The set of fixed points of a permutation $\sigma \in \mathcal{S}_\ell$ is defined as $\{i \in [\ell]: \sigma(i) = i\}$ and the exceedance of $\sigma$ is defined to be the number of indices $i$ such that $\sigma(i) > i$ (for more details, see Remark~\ref{rem:exc} and Lemma~\ref{lem:derangements}). Let $\mathcal{S}_\ell(n, k)$ denote the subset of $\mathcal{S}_\ell$ with $\ell - n$ fixed points and $k$ exceedances. Then 
\begin{equation*}
    |\mathcal{S}_\ell(n, k)| = \binom{\ell}{n} \cdot T(n, k),
\end{equation*}
where $T(n, k)$ is the derangement number defined in Remark~\ref{rem:exc}.

Since all permutations in $\mathcal{S}_\ell \backslash \{id_\ell\}$ have at most $\ell-2$ fixed points and at least $1$ exceedance, we can further expand~\eqref{eq:d1} as
\begin{equation}
    \perm(B_\ell) = 1 +  \sum_{n=2}^\ell \; \sum_{k=1}^{n-1} \sum_{\sigma \in \mathcal{S}_\ell(n, k)} \prod_{i=1}^\ell b_{i, \sigma(i)}  \,. \label{eq:pfmain}
\end{equation}

For a permutation $\sigma \in \mathcal{S}_\ell(n, k)$,  
\begin{equation*}
    \prod_{i=1}^\ell b_{i, \sigma(i)} = \prod_{i > \sigma(i)} b_{i, \sigma(i)} \cdot \prod_{i < \sigma(i)} b_{i, \sigma(i)}  \leq  \prod_{i > \sigma(i)} 2 \tilde{f}(i-\sigma(i)) \cdot \prod_{i < \sigma(i)} \frac{2\tilde{f}(\ell-\sigma(i)+i)}{\tilde{f}(\ell)}\,.
\end{equation*}
where the last inequality follows from the hypothesis of the Lemma.

Since $\sum_{i=1}^\ell i-\sigma(i) = 0$ for any permutation $\sigma$,  $\sum_{i > \sigma(i)}  i - \sigma(i) + \sum_{i < \sigma(i)} \ell - \sigma(i) + i  = \ell \cdot |\{i : \sigma(i) > i\}| = \ell \cdot k$ for any $\sigma \in \mathcal{S}_\ell(n, k)$. Therefore, if $\sigma \in \mathcal{S}_\ell(n, k)$, then there exist integers $1 \leq x_1, x_2, \ldots, x_n \leq n-1$ with $\sum_{i=1}^n x_i = \ell \cdot k$, such that
\begin{align*}
    \prod_{i=1}^\ell b_{i, \sigma(i)} \leq 2^n \cdot \frac{\prod_{i=1}^n \Tilde{f}(x_i)}{\Tilde{f}(\ell)^k}\,.
\end{align*}

Since $\tilde{f}$ satisfies $\tilde{f}(a+b) \geq \tilde{f}(a) \cdot \tilde{f}(b)$, under the constraints on $x_i$'s, for any $k \geq n/2$,
\begin{align*}
    \prod_{i=1}^\ell b_{i, \sigma(i)} \leq 2^n \cdot \frac{\Tilde{f}(\ell-1)^k \cdot \Tilde{f}(2k-n+1) \cdot \Tilde{f}(1)^{n-k-1}}{\Tilde{f}(\ell)^k}\,,
\end{align*}
and for $k < n/2$,
\begin{align*}
    \prod_{i=1}^\ell b_{i, \sigma(i)} \leq 2^n \cdot \frac{\Tilde{f}(\ell-1)^{k-1} \cdot \Tilde{f}(\ell-n+2k-1) \cdot \Tilde{f}(1)^{n-k}}{\Tilde{f}(\ell)^k} \,.
\end{align*}
Using the definition of $\Tilde{f}$, for any $k \geq n/2$,
\begin{align}
    \prod_{i=1}^\ell b_{i, \sigma(i)} \leq 2^{2n-k-1} \cdot \frac{((2k-n+1)!)^{11}}{\ell^{11k}} \, , \label{eq:p1}
\end{align}
and for $k < n/2$,
\begin{align}
    \prod_{i=1}^\ell b_{i, \sigma(i)} \leq 2^{2n-k} \cdot \frac{1}{\ell^{11(k-1)} \cdot (\ell \cdot (\ell-1) \ldots (\ell-n+2k))^{11} } \, .\label{eq:p2}
\end{align}

For some $k$ and $n$ with $k < n/2$, summing over all permutations in $\mathcal{S}_\ell(n,k)$ gives
\begin{align}
    \sum_{\sigma \in \mathcal{S}_\ell(n, k)} \prod_{i=1}^\ell b_{i, \sigma(i)} \leq \binom{\ell}{n} \cdot T(n,k) \cdot 2^{2n-k} \cdot \frac{1}{\ell^{11(k-1)} \cdot (\ell \cdot (\ell-1) \ldots (\ell-n+2k))^{11} } \, . \label{eq:p5}
\end{align}
We will bound the three terms of equation~\eqref{eq:p5}, namely $\binom{\ell}{n}$, $T(n,k)$, and $2^{2n-k}$ separately.

Expanding the first term, we get 
\begin{align}
   \binom{\ell}{n} \cdot \frac{1}{\ell^{(k-1)} \cdot (\ell \cdot (\ell-1) \ldots (\ell-n+2k))} = \frac{\ell \cdot (\ell-1) \ldots (\ell-n+1)}{n! \cdot \ell^{(k-1)} \cdot (\ell \cdot (\ell-1) \ldots (\ell-n+2k))} < \frac{1}{n! \cdot \ell^{k-1}} \,. \label{eq:p4}
\end{align}
For the third term, note that $k< n/2$ implies that $2^{2n-k} < 2^{3n-3k}$, and since $k \geq 1$, $\ell -n +2k > 2$. Using these two facts, we get 
\begin{equation}
     2^{2n-k} \cdot \frac{1}{\ell^{3(k-1)} \cdot (\ell \cdot (\ell-1) \ldots (\ell-n+2k))^{3} } < \frac{2^{3(n-k)}}{(\ell - n + 2k)^{3(n-k)}} < 1 \,. \label{eq:p3}
\end{equation} 

Plugging~\eqref{eq:p4} and~\eqref{eq:p3} in~\eqref{eq:p5}, we get
\begin{align}
    \sum_{\sigma \in \mathcal{S}_\ell(n, k)} \prod_{i=1}^\ell b_{i, \sigma(i)} \leq \frac{1}{n!\cdot \ell^{k-1}} \cdot T(n,k)  \cdot \frac{1}{\ell^{7(k-1)} \cdot (\ell \cdot (\ell-1) \ldots (\ell-n+2k))^{7} } \, .\label{eq:p6}
\end{align}

For $k = 1$, $T(n, k) = 1$ and therefore \begin{align}
    \sum_{\sigma \in \mathcal{S}_\ell(n, k)} \prod_{i=1}^\ell b_{i, \sigma(i)} &\leq  \frac{1}{n!} \cdot  \frac{1}{(\ell \cdot (\ell-1) \ldots (\ell-n+2))^{7}}\,. \label{eq:pf1}
\end{align}

For any $2 \leq k < n/2$,  using Lemma \ref{lem:derangements}, we have $T(n, k) \leq (2k+3)^{n}$. Since $k \geq 2$, $2k+3 \leq 2 \cdot 2k$, and therefore $T(n, k) \leq (2 \cdot 2k)^{n+2}$. Plugging this is~\eqref{eq:p6}, we get 
\begin{align}
    \sum_{\sigma \in \mathcal{S}_\ell(n, k)} \prod_{i=1}^\ell b_{i, \sigma(i)} &\leq \frac{1}{n! \cdot \ell^{k-1}} \cdot (2k)^{n} \cdot 2^{n} \cdot \frac{1}{\ell^{7(k-1)} \cdot (\ell \cdot (\ell-1) \ldots (\ell-n+2k))^{7} }\,. \label{eq:p7} 
\end{align}

Moreover $k < n/2$ implies that $n-k > n/2$, and as a result  $n \leq 2(n-k)$ and $ 2^{n} \cdot (2k)^{n} \leq 2^{2(n-k)} \cdot (2k)^{2(n-k)}$. Therefore,
\begin{align}
   2^{n} \cdot (2k)^{n}  &\cdot \frac{1}{\ell^{4(k-1)} \cdot (\ell \cdot (\ell-1) \ldots (\ell-n+2k))^4} \notag \\ &<  2^{2(n-k)} \cdot (2k)^{2(n-k)}  \cdot \frac{1}{\ell^{4(k-1)} \cdot (\ell \cdot (\ell-1) \ldots (\ell-n+2k))^4} \notag \\
    &< \frac{(2)^{2(n-k)}}{(\ell-n+2k)^{2(n-k)}}   \cdot \frac{(2k)^{2(n-k)}}{(\ell-n+2k)^{2(n-k)}} < 1, \label{eq:p8}
\end{align}
where the last inequality follows from  $2k \leq \ell -n + 2k$ and $2 \leq \ell -n + 2k$.

Combining~\eqref{eq:p7} and~\eqref{eq:p8}, we have for any $2 \leq k < n/2$,
\begin{equation}
     \sum_{\sigma \in \mathcal{S}_\ell(n, k)} \prod_{i=1}^\ell b_{i, \sigma(i)} \leq \frac{1}{n! \cdot  \ell^{k-1}} \cdot \frac{1}{\ell^{3(k-1)} \cdot (\ell \cdot (\ell-1) \ldots (\ell-n+2k))^3} \,. \label{eq:pf2}
\end{equation}

For a fixed $k$ and $n$ with $k \geq n/2$, summing over all permutations in $\mathcal{S}_\ell(n,k)$,
\begin{align}
    \sum_{\sigma \in \mathcal{S}_\ell(n, k)} \prod_{i=1}^\ell b_{i, \sigma(i)} &\leq \binom{\ell}{n} \cdot T(n,k) \cdot 2^{2n-k-1} \cdot \frac{((2k-n+1)!)^{11}}{\ell^{11k}}\,. \label{eq:p11}
\end{align}
We will again bound the three terms, namely $\binom{\ell}{n}$, $T(n,k)$, and $2^{2n-k-1}$ separately.

For the first term, since $n/2 \leq k$, $2n-k-1 \leq 3k-1$, and therefore
\begin{align*}
  2^{2n-k-1}  \cdot \frac{((2k-n+1)!)^3}{\ell^{3k}} &\leq  2^{3k-1}  \cdot  \frac{((2k-n+1)!)^3}{\ell^{3k}}\,.
\end{align*}
Since $2k-n+1 \leq k$, and $k+1 \leq \ell$,
\begin{align*}
  2^{3k-1}  \cdot \frac{((2k-n+1)!)^3}{\ell^{3k}} &\leq  \frac{1}{2}  \cdot  \frac{(2^k \, k!)^3}{(k+1)^{3k}} \,.
\end{align*}
For $k = 1, 2, 3, 4, 5$, $\frac{(2^k k!)^3}{(k+1)^{3k}} \leq 1$. For $k \geq 6$, $2^k k! \leq k^k$. Therefore, 
\begin{align}
  2^{2n-k-1}  \cdot \frac{((2k-n+1)!)^3}{\ell^{3k}} &\leq  \frac{1}{2} \,. \label{eq:p10}
\end{align}

Expanding the third term, 
\begin{align*}
   \binom{\ell}{n} \cdot \frac{((2k-n+1)!)^2}{\ell^{2k}} &= \frac{\ell\cdot (\ell-1) \ldots (\ell-n+1)}{n!} \cdot \frac{((2k-n+1)!)^2}{\ell^{2k}}\\
   &\leq \frac{1}{n!} \cdot \frac{((2k-n+1)!)^2}{\ell^{2k-n}} = \frac{1}{\ell (n-1)!} \cdot \frac{((2k-n+1)!)^2}{n\ell^{2k-n-1}} \leq \frac{(2k-n+1)!}{\ell (n-1)!} \,.
\end{align*}
Since $k+1 \leq n$, we have $2k-n+1 \leq n-1$, and therefore
\begin{align}
   \binom{\ell}{n} \cdot \frac{((2k-n+1)!)^2}{\ell^{2k}} &\leq \frac{(2k-n+1)!}{l(n-1)!} \leq \frac{1}{\ell} \,. \label{eq:p9}
\end{align}

Plugging in~\eqref{eq:p10} and~\eqref{eq:p9} in~\eqref{eq:p11}, 
\begin{align}
    \sum_{\sigma \in \mathcal{S}_\ell(n, k)} \prod_{i=1}^\ell b_{i, \sigma(i)} &\leq  \frac{1}{2\ell} \cdot T(n,k) \cdot \frac{((2k-n+1)!)^6}{\ell^{6k}} \,. \label{eq:p12}
\end{align}
Since $T(n, n-1) = 1$, for $k = n-1$, we have 
\begin{align}
    \sum_{\sigma \in \mathcal{S}_\ell(n, k)} \prod_{i=1}^\ell b_{i, \sigma(i)} &\leq  \frac{1}{2\ell} \cdot \frac{((n-1)!)^6}{\ell^{6(n-1)}} \,. \label{eq:pf4}
\end{align}
By Lemma \ref{lem:derangements}, for $k < n-1$, $T(n, k) \leq  (2n-2k+5)^{n}$ for $k \geq n/2$, and
\begin{align*}
    \sum_{\sigma \in \mathcal{S}_\ell(n, k)} \prod_{i=1}^\ell b_{i, \sigma(i)} &\leq  \frac{1}{2\ell} \cdot (2n-2k+5)^{n} \cdot \frac{((2k-n+1)!)^6}{\ell^{6k}}\,.
\end{align*}
Let $n = 2k - z$, then
\begin{align*}
  (2n-2k+5)^{n} &\cdot \frac{((2k-n+1)!)^6}{\ell^{6k}} = (2k-2z+5)^{2k-z}  \cdot \frac{((z+1)!)^6}{\ell^{6k}} \\
  &\leq (2k-z+5)^{2k-z}  \cdot \frac{((z+1)!)^6}{\ell^{6k}} \,.
\end{align*}
Taking derivative of $ \frac{(2k-z+7)^{2k-z+2}}{\ell^{6k}}$ with respect to $k$,
\begin{align*}
   \frac{(2k-z+5)^{2k-z}}{\ell^{6k}}\left(2\cdot \log(2k-z+5) + 2\cdot\frac{2k-z}{2k-z+5} -6\log(\ell)\right) \leq 0 \,.
\end{align*}
Therefore $ \frac{(2k-z+5)^{2k-z+2}}{\ell^{6k}}$ is a non-increasing function of $k$. Since $n = 2k-z \geq 1$, $k$ satisfies $2k \geq z+1$. So $ \frac{(2k-z+5)^{2k-z}}{\ell^{6k}}$ is maximized when $2k = z + 1$. Therefore,
\begin{align*}
  (2k-z+5)^{2k-z}  \cdot \frac{((z+1)!)^6}{\ell^{6k}} 
  &\leq 6 \cdot \frac{((z+1)!)^6}{\ell^{6k}} = 6 \cdot \frac{((2k-n+1)!)^6}{\ell^{6k}} \leq  6 \cdot \frac{1}{\ell^{6(n-k-1)}} \,,
\end{align*}
where the last inequality follows from $2k-n+1 \leq \ell$.
Plugging this bound in~\eqref{eq:p12} gives
\begin{align}
    \sum_{\sigma \in \mathcal{S}_\ell(n, k)} \prod_{i=1}^\ell b_{i, \sigma(i)} &\leq  \frac{1}{2\ell} \cdot \frac{6}{\ell^{6(n-k-1)}} \,. \label{eq:pf3}
\end{align}

Plugging in~\eqref{eq:pf1},~\eqref{eq:pf2},~\eqref{eq:pf4}, and~\eqref{eq:pf3} into ~\eqref{eq:pfmain}, 
\begingroup
\allowdisplaybreaks
\begin{align*}
  \perm(B_\ell) 
    = &1+ \sum_{n=2}^\ell \sum_{k=1}^{n-1} \sum_{\sigma \in \mathcal{S}_\ell(n, k)} \prod_{i=1}^\ell b_{i, \sigma(i)} \\
    = &1+ \sum_{n=2}^\ell \sum_{\sigma \in \mathcal{S}_\ell(n, 1)}\prod_{i=1}^\ell b_{i, \sigma(i)}  + \sum_{n=5}^\ell \sum_{2\leq k < n/2} \sum_{\sigma \in \mathcal{S}_\ell(n, k)}\prod_{i=1}^\ell b_{i, \sigma(i)} \\
    &+ \sum_{n=2}^\ell \sum_{\sigma \in \mathcal{S}_\ell(n, n-1)}\prod_{i=1}^\ell b_{i, \sigma(i)}+\sum_{n=2}^\ell \sum_{ k = \lceil n/2\rceil}^{n-2} \sum_{\sigma \in \mathcal{S}_\ell(n, k)}\prod_{i=1}^\ell b_{i, \sigma(i)} \\
    \leq &1+ \sum_{n=2}^\ell \frac{1}{n!} \cdot  \frac{1}{(\ell \cdot (\ell-1) \ldots (\ell-n+2))^{7}} \\
    &+   \sum_{n=5}^\ell \sum_{2 \leq k < n/2} \frac{1}{n! \cdot \ell^{4(k-1)} \cdot (\ell \cdot (\ell-1) \ldots (\ell-n+2k))^3}  \\
    &+ \sum_{n=2}^\ell \frac{1}{2\ell} \cdot \frac{((n-1)!)^6}{\ell^{6(n-1)}} + \sum_{n=2}^\ell \sum_{k \geq n/2}  \frac{1}{2\ell} \cdot \frac{6}{\ell^{6(n-k-1)}}\\
    \leq &1+\frac{1}{\ell^5} + \frac{1}{\ell^2}  + \frac{1}{2\ell^6}+ \sum_{n=2}^\ell \frac{3}{\ell} \cdot \frac{1}{\ell^6 -1} \\
    \leq &1+\frac{1}{\ell^5} + \frac{1}{\ell^2}  + \frac{1}{2 \ell^6}+ \frac{3}{\ell^6 -1} \leq 1.13\,,
\end{align*}
for $\ell \geq 3$.
\end{proof}
\endgroup

\begin{remark}[Exceedances and the Eulerian Number] \label{rem:exc} For a permutation $\sigma \in \mathcal{S}_n$,
\begin{itemize}
    \item The exceedance of $\sigma$ is defined as $|\{i \in [n-1]: \sigma(i) > i\}|$.
    \item The Eulerian number $E(n, k)$ is defined to be the number of permutations in $\mathcal{S}_n$ with $k-1$ exceedances.
    \item The derangement number $T(n, k)$ is defined to be the number of derangements in $\mathcal{S}_n$ with $k$ exceedances. 
\end{itemize}
The explicit formula for $E(n, k)$ is  $ E(n,k)=\sum_{j=0}^{k+1}(-1)^{j}{\binom {n+1}{j}} \cdot (k+1-j)^{n}$ (page 273, \cite{comtet1974advanced}). The exponential generating function of $E(n,k)$ is given by (page 273, \cite{comtet1974advanced})
\begin{equation}
    \sum_{n=0}^{\infty} \sum_{k=0}^n E(n,k) \; t^{k} \; \frac{x^n}{n!} = \frac{t-1}{t-e^{(t-1)x}} \,. \label{eq:eu}
\end{equation}

The exponential generating function of $T(n,k)$ is given by (Proposition 5, \cite{brenti1990unimodal})
\begin{equation}
    \sum_{n=0}^{\infty} \sum_{k=1}^{n-1} T(n,k) \; t^{k} \; \frac{x^n}{n!} = e^{-t} \cdot\frac{t-1}{t-e^{(t-1)x}} \,. \label{eq:dr}
\end{equation}

Comparing~\eqref{eq:eu} and~\eqref{eq:dr}, we infer
\begin{equation}
    T(n,k) = \sum_{j=0}^{n} (-1)^{(n-j)} \cdot \binom{n}{j} \cdot E(j, k) \,. \label{eq:derangement}
\end{equation} 
\end{remark}
\begin{lemma} \label{lem:derangements}
For any positive integer $n$, 
\begin{itemize}
    \item $T(n,1) = T(n,n-1) = 1$, 
    \item $T(n, k) = T(n, n+1-k) < (2k+3)^{n}$ for any $k \in \{2, \ldots, n-2\}$.
\end{itemize}
\end{lemma}
\begin{proof}
Since $ E(n,m)=\sum _{k=0}^{m+1}(-1)^{k}{\binom {n+1}{k}} \cdot (m+1-k)^{n}$, taking absolute values, we get
\begin{align*}
    E(n, m) &< \binom{n+1}{0} \cdot (m+1)^{n} + \binom{n+1}{2} \cdot (m+1-2)^{n} + \ldots \\
    &< (m+1)^n \cdot \left(\sum_{i = 0}^n \binom{n+1}{2i}  \right) \leq (m+1)^n \cdot 2^n.
\end{align*}
Using equation~\ref{eq:derangement} and taking absolute values, we get
\begin{align*}
    T(n,k) &< \sum_{j=0}^{n}  \binom{n}{j} \cdot E(j, k) < \sum_{j=0}^{n}  \binom{n}{j} \cdot  2^j \cdot (k+1)^j = (1 + 2(k+1))^{n}.
\end{align*}
\end{proof}

The following lemma is an extension to \cref{lem:vol-determinant}.
\begin{lemma}\label{vol-determinant2}
Let $S$ be a subset of $r < d$ vectors in $\R^d$. Let $C$ be a cycle in $G(S)$, and let $X = C \backslash S$ and $Y = S\backslash C$ such that $|X| = |Y| = \ell$. Let $X = Y A_C + X_{\perp}$, where $X_\perp$ denotes the projection of vectors in $X$ orthogonal to $\Span(Y)$. Then the change in objective value is given by
	\[
	\frac{\vol(S\Delta C)^2}{\vol(S)^2}
	\geq \det(A_C^\top A_C). 
	\]
	
	In particular, if $C'$ is the projection of the vectors in $C$ onto $\Span(S)$, then $\vol(S\Delta C)^2 \geq \vol(S\Delta C')^2$.
\end{lemma}

\begin{proof}
Let $ W = S \del Y$. Let us abuse notation to denote by $X$ the matrix whose columns are the vectors in set $X$ and similarly for other sets defined above. Then let $X = YA_C + WA' + Z$, where $Z$ is the component of $X$ which is orthogonal to $\Span(S)$. Let $Y = Y_\perp + Y_W$, where $Y_\perp$ is the component of $Y$ orthogonal to $\Span(W)$.
	
Let $M$ denote the matrix whose columns are the elements of $S \Delta C$. Concretely, $M = [X \,\, W]$ following the notation above.
	Note that
	\begin{align}
		\vol(S\Delta C)^2
		&= \det(M^\top M)
		= \det\left( \bmat{X^\top \\W^\top } \bmat{X&W} \right) = \det\left( \bmat{X^\top X & X^\top W \\ W^\top X & W^\top W} \right) \notag \\
		&= \det\left(W^\top W\right)^{1/2} \det\left(X^\top X - X^\top W(W^\top W)^{-1}W^\top X\right),\label{eq:e4}
	\end{align}
	where the last equation follows from the fact that $M^\top M$ is positive definite.
	
	Now, since $Z$ is orthogonal to both $Y$ and $W$, we get
	\begin{equation}
	    X^\top X = (YA_C)^\top (YA_C) + (WA')^\top (WA') + Z^\top Z + (YA_C)^\top (WA') + (WA')^\top (YA_C). \label{eq:e1}
	\end{equation} 
	Additionally, $W(W^\top W)^{-1}W^\top $ is the projection matrix on the column space of $W$, so
	\begin{equation}
	     W(W^\top W)^{-1}W^\top X = Y_WA_C + WA'\,. \label{eq:e2}
	\end{equation} 
	Multiplying equation~\eqref{eq:e2} with $X$ gives
 	\begin{equation}
	    X^\top W(W^\top W)^{-1}W^\top X = (YA_C)^\top (Y_WA_C) + (YA_C)^\top (WA') + (WA')^\top (Y_WA_C) + (WA')^\top (WA').  \label{eq:e3}
	\end{equation} 
	Subtracting equation~\eqref{eq:e3} from equation~\eqref{eq:e1}, we see that
	\begin{align*}
		X^\top X - X^\top W(W^\top W)^{-1}W^\top X
		&= (YA_C)^\top (Y - Y_W)A_C + (WA')^\top (Y - Y_W)A_C + Z^\top Z\\
		&= (Y_\perp A_C)^\top (Y_\perp A_C) + Z^\top Z.
	\end{align*}
	Substituting the value of $X^\top X - X^\top W(W^\top W)^{-1}W^\top X$ from the above equation in equation~\eqref{eq:e4}, we conclude that
	\begin{equation*}
	    \vol(S\Delta C)^2 = \det\left(W^\top W\right)\det\left((Y_\perp A_C)^\top (Y_\perp A_C) + Z^\top Z\right).
	\end{equation*} 
	
	Similarly,
	\begin{align}
		\vol(S)^2
		&= \det\left( \bmat{Y^\top \\W^\top }\bmat{Y&W} \right) = \det\left( \bmat{Y^\top Y & Y^\top W\\W^\top Y & W^\top W} \right) \notag \\
		&= \det\left( W^\top W\right)^{1/2} \det\left( Y^\top Y - Y^\top W(W^\top W)^{-1}W^\top Y \right) \notag \\
		&= \det\left( W^\top W\right)^{1/2}\det\left( Y^\top Y - Y^\top Y_W \right) \\
		&= \det\left( W^\top W\right)^{1/2} \det\left( Y_\perp^\top Y_\perp\right). \label{eq:e5}
	\end{align}
	Finally, dividing equation~\eqref{eq:e3} by equation~\eqref{eq:5} gives
	\begin{equation*}
	    \frac{\vol(S\Delta C)^2}{\vol(S)^2} = \frac{\det\left((Y_\perp A_C)^\top (Y_\perp A_C) + Z^\top Z\right)}{\det\left( Y_\perp^\top Y_\perp\right)}.
	\end{equation*} 
	
	Since $Z^\top Z$ is positive semidefinite, we conclude that
	\begin{align*}
		\frac{\vol(S\Delta C)^2}{\vol(S)^2}
		&= \frac{\det\left((Y_\perp A_C)^\top (Y_\perp A_C) + Z^\top Z\right)}{\det\left( Y_\perp^\top Y_\perp\right)}  \\
		&\geq \frac{\det\left((Y_\perp A_C)^\top (Y_\perp A_C)\right)}{\det\left( Y_\perp^\top Y_\perp\right)}  = \det(A_C^\top A_C).
	\end{align*}
	If $C'$ is the projection of $C$ onto $\Span(S)$, then Lemma \ref{lem:vol-determinant} implies that $\vol(S\Delta C')^2 = \vol(S)^2\cdot \det(A_C^\top A_C) \leq \vol(S\Delta C)^2$.
\end{proof}

\end{document}